\documentclass[12pt,  oneside, a4paper]{article}
\usepackage[utf8]{inputenc}
\usepackage[a4paper, total={6.5in,9in}]{geometry}

\usepackage{hyperref}
\usepackage{amsmath}
\usepackage{mathtools}
\usepackage{lipsum}
\usepackage{blindtext}
\usepackage{titlesec}
\usepackage{graphicx}
\usepackage{amssymb}
\usepackage{amsthm}
\usepackage{rotating}
\usepackage{subcaption}
\usepackage{caption}
\usepackage{svg}
\usepackage{verbatim}
\usepackage{enumerate}
\usepackage{enumitem}
\usepackage{subfloat}
\usepackage[square,sort]{natbib}
\usepackage{url}
\usepackage{cleveref}

\newtheorem{theorem}{Theorem}[section]
\newtheorem*{theorem*}{Theorem}

\newtheorem*{lemma*}{Lemma}
\newtheorem{claim}[theorem]{Claim}

\newtheorem{corr}[theorem]{Corollary}

\newtheorem{example}[theorem]{Example}

\newtheorem{aspt}{Assumption}[section]
\newtheorem{defn}{Definition}[section]

\newcommand{\footremember}[2]{%
    \footnote{#2}
    \newcounter{#1}
    \setcounter{#1}{\value{footnote}}%
}

\title{On Rider Strategic Behavior in Ride-Sharing Platforms}
\author{Jay Mulay\footremember{alley}{Georgia Institute of Technology. Email: jay.mulay@gatech.edu} 
\and Diptangshu Sen\footremember{trailer}{Georgia Institute of Technology. Email: dsen30@gatech.edu}
\and Juba Ziani\footremember{somethingelse}{Georgia Institute of Technology. Email: jziani3@gatech.edu}
}
\date{\today}

\begin{document}

\maketitle

\begin{abstract}
Over the past decade, ride-sharing services have become increasingly important, with U.S. market leaders such as Uber and Lyft expanding to over 900 cities worldwide and facilitating billions of rides annually. This rise reflects their ability to meet users' convenience, efficiency, and affordability needs. However, in busy areas and surge zones, the benefits of these platforms can diminish, prompting riders to relocate to cheaper, more convenient locations or seek alternative transportation.

While much research has focused on the strategic behavior of drivers, the strategic actions of riders, especially when it comes to riders \emph{walking} outside of surge zones, remain under-explored. This paper examines the impact of rider-side strategic behavior on surge dynamics. We investigate how riders' actions influence market dynamics, including supply, demand, and pricing. We show significant impacts, such as spillover effects where demand increases in areas adjacent to surge zones and prices surge in nearby areas. Our theoretical insights and experimental results highlight that rider strategic behavior helps redistribute demand, reduce surge prices, and clear demand in a more balanced way across zones.
\end{abstract}

\section{Introduction}\label{sec:intro}
Over the past decade, ride-sharing services have grown exponentially in importance. U.S. market leaders such as Uber and Lyft have now expanded their operations to over 900 cities worldwide, facilitating billions of rides annually. The rise of ride-sharing platforms reflects how they have adeptly addressed users' needs for convenience, efficiency, and reliability. One of the major reasons for the popularity of ride-sharing platforms is their ability to provide users with rides from virtually anywhere, at any time, at a relatively affordable price. However, the benefits of ride-sharing platforms can diminish in busy areas and surge zones, incentivizing strategic behavior by riders who may relocate to cheaper and more convenient locations or seek alternate means of transportation entirely.

Much academic work on ride-sharing has recognized the prevalence of strategic behavior among drivers within these platforms. Drivers can strategically decide which locations to serve or when to be online to maximize their earnings during surges; they can turn off the application while transiting to avoid taking rides on their way to more profitable locations; they can selectively accept rides based on profit and convenience, and even cancel rides when needed. While there is a significant focus on driver behavior in the literature, related work has largely neglected rider-side behavior. In particular, riders often strategically relocate to avoid surge zones and to secure cheaper rides during major events such as sports games or concerts.

The goal of this paper is to study the impact of rider-side strategic behavior on surge dynamics. We aim to understand how such rider behavior influences market dynamics, particularly the evolution of demand inside and outside the surge zone. This paper highlights several significant impacts: spillover effects where demand becomes higher than normal and price surges in areas adjacent to surge zones. We provide precise theoretical insights and experimental results to elucidate and quantify these effects. Our analysis reveals that rider strategic behavior helps redistribute and smooth out the demand across different areas, thereby reducing the price gap across the surge boundary.

% \jz{Keep the last paragraph in mind when writing. We want to make our results consistent with this/vice-versa this consistent with the results.}

\noindent\textbf{Summary of Contributions:} The main contributions of our work are as follows:
\begin{itemize}
    \item We introduce a new model highlighting the phenomenon of strategic riders on ride-sharing platforms in Section~\ref{sec:model}. In our model, riders can be waiting in one of two zones: a surge and a non-surge zone. Riders faced with high prices during a surge can choose to walk out of the surge zone in search of cheaper rides. We capture, in our initial model, the rate of riders moving across the surge boundary through a simple functional form that depends on the difference in demands between the surge and the non-surge areas. We extend this model to a more realistic agent-based and game-theoretic model in Section~\ref{sec:agent_model}.
    \item In Section~\ref{sec:dynamics}, we analyze the dynamics of our initial theoretical model. We first characterize the shape and convergence properties of demand curves in both surge and non-surge zones (Theorems \ref{thm:Ds_conv} and \ref{thm:Dns_conv}). We also show the existence of two types of surges: i) a \emph{localized} type of surge where the surge demand does not accumulate outside the surge boundary (it remains localized as the name suggests) and ii) \emph{spill-over} type of surge where the surge demand spills over and accumulates in the non-surge zone causing the demand there to increase.
    \item Finally, in Section~\ref{sec:experiments}, we develop a much more comprehensive game-theoretic, agent-based simulation model of a ride-sharing platform where each rider is modeled as a rational agent. Agents decide to move to the non-surge zone if they get an increased utility from doing so, where this utility depends on prices in each zone and the cost of moving. This model also incorporates additional considerations to take the model one step closer to reality, including strategic riders that may prefer going to the more expensive surge zone, and stochasticity in arrival rates. 

    Our agent-based model serves two purposes. First, it shows that our simple initial, theoretical model that solely aimed to capture the impact of strategic rider behavior still provides a good approximation to the more complex and realistic agent-based model, highlighting its value as a simple and tractable alternative. Second, it highlights the role of strategic riders in surge dynamics, in particular quantifying the extent to which this strategic behavior helps smooth out the demand and prices across zones. 
\end{itemize}

\subsection{Related Work} 

There is a rich body of literature on challenges faced by ride-sharing platforms in different settings:~\cite{agatz2012optimization,furuhata2013ridesharing, kooti2017analyzing,chan2012ridesharing,jin2018ridesourcing} to only name a few.

A prominent line of work involves studying pricing mechanisms for ridesharing platforms and in particular, i) static versus dynamic pricing models~\cite{banerjee2015pricing, pandit2019pricing}, and ii) spatial, temporal, and spatio-temporal pricing strategies~\cite{bimpikis2019spatial, ma2022spatio}. Differential pricing strategies have also been explored, with an example being the work of \citet{zhao2022differential}. Apart from pricing, there is also a significant amount of work that has studied the general operations management of ride-hailing systems, including empty vehicle rerouting, dispatch, matching and joint optimization ~\cite {braverman2019empty,chen2020optimization,alonso2017demand,yan2020dynamic}. Recently, researchers have also considered how ride-sharing platforms can encourage riders in congested areas to walk to nearby locations for pickup as it has been shown to improve the service efficiency by reducing pickup times for drivers ~\cite{yao2021new,fielbaum2021demand}.  

Our work is closely related to the two lines of research. The first one studies \emph{surge} pricing, when rider demand vastly outdoes the supply of drivers. Surge pricing aims to better match supply with demand during periods of high demand and lower prices when demand decreases; for example,~\cite{castillo2017surge} show how low supply relative to demand can lead to market failures and ``wild goose chase'' phenomena where free drivers are thinly spread and sent to far away locations, hurting both riders' waiting times and drivers profits. They highlight the role of surge pricing as a solution to this problem. There have also been other works like ~\cite{castillo2023benefits} which study the overall welfare effects of surge pricing. Second, a recent line of work recognizes user incentives and strategic behavior in ride-sharing platforms, particularly driver-side. Drivers employ varied strategies to maximize profits and utility, including choosing which rides to accept, chasing surges, and strategically signing in and out to manipulate supply and prices, to only name a few. There has been increased focus on modeling and understanding this strategic driver behavior, such as \cite{ashkrof2020understanding}. \citet{garg2022driver} and \citet{cashore2022dynamic} study the design of incentive-compatible pricing mechanisms for drivers, while \citet{rheingans2019ridesharing} study pricing that accounts for drivers' location preferences. \citet{tripathy2022driver, yu2022price} focus on how coalitions of drivers can collude to manipulate the system. There have also been attempts to understand how such strategic behavior may affect the marketplace. For example, ~\cite{besbes2021surge} investigates the supply side's spatial response to surge pricing when drivers are strategic. They show that optimal surge pricing induces a spatial decomposition of the supply - in some locations supply is exactly matched with demand, while in other locations supply may either be deliberately over-congested or completely diverted away. Unlike the driver side, there has been limited effort to account for riders being strategic. ~\citet{chen2020pricing} is an exception in this regard, it studies pricing policies for ride-sharing platforms when both drivers and riders are strategic and forward-looking, particularly riders can monitor prices and decide when to accept a ride.   

The work which is perhaps most closely related to ours is the recent contribution of ~\citet{hu2022surge}. Similar to our work combining surge dynamics and rider strategic behavior,~\citet{hu2022surge} examine two-sided temporal dynamics and show that drivers respond more slowly to surge pricing compared to riders. They propose two types of equilibrium surge pricing policies and characterize the responses of both drivers and riders. In their paper, the main strategic aspect that is modelled for riders is that they can choose to wait for better prices. Our work differs significantly from theirs in the sense that we consider a \emph{different} form of strategic behavior on the rider-side where they can choose to walk away from the surge zone in search of cheaper rides. We also characterize the effect of this strategic behavior on the surge dynamics itself - something the previous paper does not consider.

\section{Model}\label{sec:model}
\begin{figure}[!ht]
      \centering
      \includegraphics[width=0.7\textwidth]{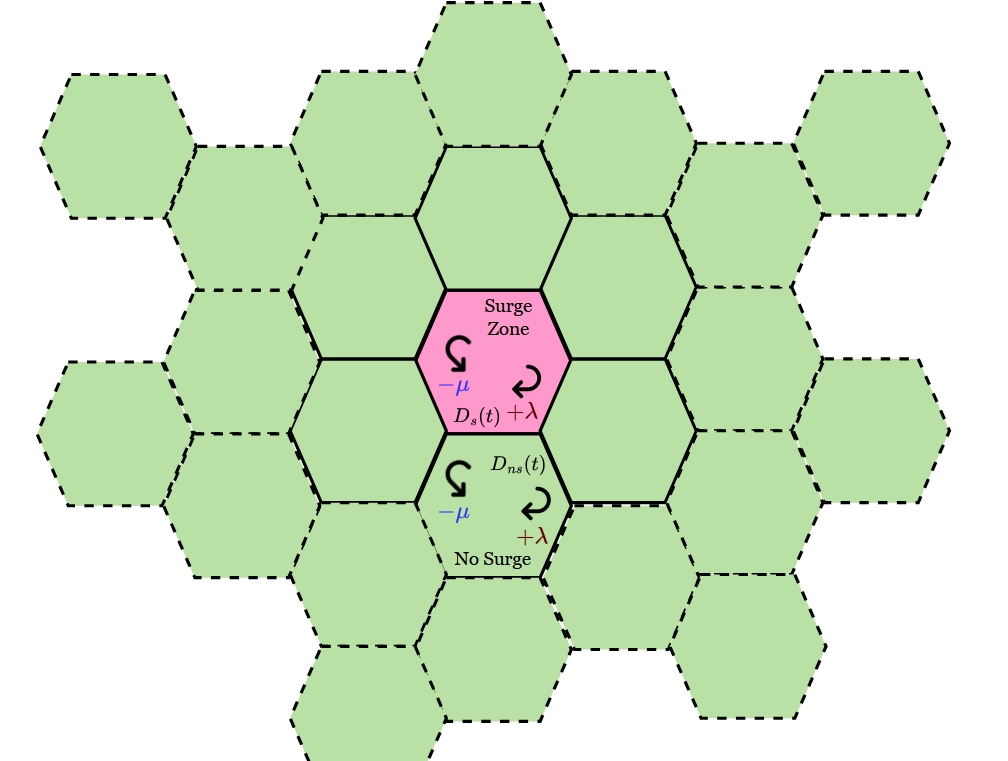}
      \caption{The above map schematic represents a geography serviced by a ride-sharing platform. The entire region is divided into hexagonal blocks and the price in each block depends on the market conditions ``locally". Ride-sharing platforms like Uber employ surge pricing algorithms at a ``hyperlocal" level~\cite{uber_pricing}, and Uber in particular uses such hexagonal regions for pricing ~\cite{uber_hex}. The red block represents a region with abnormally high demand, leading to a surge in prices. We call this zone the ``surge zone". In all other surrounding blocks (marked in green), demand reflects usual levels; we call them the ``non-surge zone".}
\end{figure}

We study a model of rider dynamics at surge time on ride-sharing platforms where the riders can be \emph{strategic} and are willing to walk outside of the surge zone to call a ride. This is motivated by behavior observed in practice at large events such as concerts and sports games and backed by empirical research ~\cite{npr}. The paper takes the point of view of the relatively short time horizon that it takes for a surge to happen and clear (minutes to hours). The main goal of the paper is then to understand ride-sharing dynamics over this type of short time scales, with a focus on understanding how demand clears in surge areas during major events. 

To model these dynamics, we study a geography consisting of two main regions: a central region where a major event is happening and where the demand for rides is unusually high, called the \emph{surge zone}; and the surrounding region, initially facing typical demand levels, called the \emph{non-surge zone}. The behavior of drivers (supply side) and riders (demand side) on the platform is formalized below.

\paragraph{Driver's arrivals and behavior:} We consider \emph{non-strategic} drivers that arrive in each zone at an equal, constant, deterministic rate of $\mu$. We \emph{initially} focus on non-strategic drivers to isolate and highlight the role of \emph{rider-side} strategic behavior, which is the main focus of this paper. However, we note that the assumption that rates are equal across zones and that drivers are not strategic, used in the theoretical analysis of Section~\ref{sec:dynamics}, is later relaxed in the experiments of Section~\ref{sec:experiments}. 

\paragraph{Rider's arrivals and behavior:} At time $0$, we consider there is an initial demand in both the surge and the non-surge zone. We denote by $D_{s}(t)$ and $D_{ns}(t)$ the \emph{unmet} demands at the end of time step $t$ in the surge and non-surge zones respectively, and assume throughout the paper that $D_s(0) \gg D_{ns}(0)$. Starting in the first time step, new riders arrive in each zone at a constant, deterministic rate of $\lambda$. We assume\footnote{Otherwise, demand can never be cleared and builds up over time, leading to a trivial/uninteresting case.} that $\lambda < \mu$.

Riders can leave each zone via two mechanisms: i) a rider gets matched to a driver and gets a ride; or, ii) a rider can be \emph{strategic} ~\cite{tabloid} (a departure from much of the related work in ride-sharing), and decide \emph{to leave the surge zone in favor of the non-surge zone}, for example when the demand in the surge zone is significantly higher than the demand in the non-surge zone. We make the following assumption on rider behavior in case ii): 

\begin{aspt}\label{as:strat_riders}
The fraction of the unmet rider demand $D_s(t)$ that move from the surge to the non-surge zone at every time step is given by $f \left(D_{s}(t) - D_{ns}(t)\right)$, for a known function $f$. Further, $f(.)$ is non-decreasing, in $[0,1]$, and is positive only if $D_{s}(t) > D_{ns}(t)$.
\end{aspt}

The above condition abstracts away riders' strategic behavior by providing a functional form on the \emph{rate} at which riders decide to move strategically from the surge to the non-surge zone. First, we assume that $f$ is non-decreasing; this is a natural assumption, as the higher the difference in demand and how crowded the surge and non-surge zones are, the more riders are willing to move. Second, we note that the rate is always non-negative: this models the fact that riders would sensibly walk from high-density areas like a stadium or a concert venue to lower-density areas, but not the other way around (e.g., a rider in a busy downtown area walking closer to a major venue or stadium experiencing a surge) may be unrealistic. Second, we note that our functional form makes a \emph{linearity} assumption, where the rate of movement is a function of the difference between demands in the surge and non-surge zones. We provide a motivating explanation of this assumption below:

\begin{example}\label{ex:1}
Assume a simple rider behavior where they make strategic decisions based on price multipliers set by the platform. Denote $m_s(t)$ and $m_{ns}(t)$ respectively the price multipliers in the surge and non-surge zones. Suppose that each agents estimates the price they pay in each zone according to the multipliers, and have a cost $c$ for moving from the surge to non-surge zone. Then the rider's utility for not moving from the surge zone is $m_s(t)$, while the utility for moving is $m_{ns}(t) + c$. In this case, a rider moves if and only if $c \leq m_s(t) - m_{ns}(t)$. When the costs $c$ are i.i.d. from a distribution with c.d.f. F(.), the rate at which riders move is $F(m_s(t) - m_{ns}(t))$. Under a simple pricing rule that is proportional to demand, we recover Assumption~\ref{as:strat_riders}: one such example of a pricing rule could be where the price multiplier depends on the supply-demand mismatch, i.e., $m_s(t) \propto \frac{D_s(t)}{\mu}$ and $m_{ns}(t) \propto \frac{D_{ns}(t)}{\mu}$.
\end{example}

Another motivating explanation is that riders make decision based on \emph{perceived} demand around them. A rider that is in a crowded area perceives that they may be a large demand differential between the surge zone and adjacent, outer areas, and may want to switch locations based on that perceived differential.

We note that the assumptions that i) the rate of rider arrivals are deterministic and ii) Assumption~\ref{as:strat_riders} on the specific functional form of riders' behavior, made for the theoretical results of Section~\ref{sec:dynamics}, are relaxed in the experiments of Section~\ref{sec:experiments} in favor of stochastic arrival rates and a game-theoretic model of riders' decisions on which zone to serve as a function of platform prices.

\section{Surge: Dynamics and Convergence}\label{sec:dynamics}

This section presents a theoretical perspective to the dynamics of surge when riders are strategic in nature. The section is organized as follows: we start by deriving the equations of dynamics from first principles in Section \ref{sub:eq_dynamics}, followed by an analysis of the shapes of surge and non-surge demands and providing bounds on the times to convergence in Section \ref{sub:shape}.  

\subsection{Equations of Dynamics}\label{sub:eq_dynamics}
In this segment, we explore the dynamics of the unmet demands inside and outside the surge boundary, given by $D_s(t)$ and $D_{ns}(t)$ respectively. We start by formally deriving the one-step equation for the evolution of $D_s(t)$ and $D_{ns}(t)$ from first principles. 

Remember that $D_s(t)$ is the unmet demand in the surge zone at the end of time period $t$. We express the unmet demand at the end of period $t+1$ recursively as a function of the unmet demands at the end of the previous time period $t$. We do this by accounting for all inflows and outflows of demand from the surge zone. 

\paragraph{Inflows:} Since riders can only leave the surge zone by being matched with a ride or choosing to relocate to the non-surge zone at the start of time $t$, a carryover demand of size $D_s(t)$ is already in the surge zone at the start of period $t+1$. The surge zone also has an exogenous arrival of riders 
at rate $\lambda$ which means that $\lambda$ riders arrive over the course of period $t+1$. Therefore, the total demand to be cleared in period $t+1$ is given by $D_s(t) + \lambda$. 

\paragraph{Outflows:} Since $D_s(t)$ number of riders had already been waiting to be matched since period $t$, a fraction of them may decide to relocate to the non-surge zone at the start of period $t+1$. As described earlier, this fraction is captured by $f(\cdot)$ which is a function of the demand differential $D_s(t)-D_{ns}(t)$ across the surge boundary. There is also a supply of new drivers at fixed rate $\mu$ in each time period, so $\mu$ rides arrive over the course of period $t+1$, get matched with rides and leave the system. Therefore, the total outflow of demand from the surge zone during period $t+1$ is at most $\mu + f(D_s(t)-D_{ns}(t))\cdot D_s(t)$. 

\paragraph{Equations of dynamics:} Hence, we can express $D_s(t+1)$ in the following recursive form: 
\begin{align}
    D_s(t+1) = \max \left[0,  D_s(t) + (\lambda - \mu) - f\left( D_s(t) - D_{ns}(t) \right) \cdot D_s(t) \right]. \label{eq:evolve_s_rev}
\end{align}
Note that the $\max\left(0, \cdot\right)$ ensures that outflow of demand cannot exceed the total demand to be cleared. 
Similarly, for $D_{ns}(t)$, we obtain:
\begin{align}
    D_{ns}(t+1) = \max \left[0, D_{ns}(t) + (\lambda - \mu) + f\left( D_s(t) - D_{ns}(t) \right) \cdot D_s(t) \right]. \label{eq:evolve_ns_rev}
\end{align}

\paragraph{Boundary Conditions:} In order to solve for $D_s(t)$ and $D_{ns}(t)$, we also need to specify the state of the system at $t = 0$ which constitute our boundary conditions. We set: 
\[
     D_s(0) = D_0; \quad \text{and} \quad D_{ns}(0) = d_0;
\]
$D_0$ is the large initial demand that arises in the surge zone at $t = 0$ (e.g., people leaving a stadium at the end of a football game), while $d_0$ is the uncleared demand in the non-surge zone at the same time. We assume that $D_0 \gg d_0$. 

The rest of the section studies the evolution of the system dynamics from Equations~\eqref{eq:evolve_s_rev} and~\eqref{eq:evolve_ns_rev} above.

\subsection{Properties of Surge and Non-surge Demands}\label{sub:shape}

Our main contribution in this segment is to characterize the shapes of the surge and non-surge demands over time in Theorems~\ref{thm:Ds_conv} and~\ref{thm:Dns_conv}. Recall that $D_0$ and $d_0$ are the initial unmet demands in the surge and non-surge zones, $\lambda$ is the rate of exogenous demand and $\mu$ is the rate of supply in each zone. We start by characterizing properties of $D_s(t)$, the unmet demand in the \emph{surge} area.

\begin{theorem}\label{thm:Ds_conv}
Let $D_s(t)$ be the unmet demand inside the surge zone at any time step $t \in \{0, 1, 2,...\}$. Then, $D_s(t)$ is always non-increasing in $t$ and must converge to $0$ in a finite number of time steps. Further, if $\tau_s$ indicates the number of time steps till convergence, then $\tau_s$ must satisfy: 
\[
       \left\lfloor \frac{D_0}{(\mu - \lambda) + D_0 \cdot f(D_0)} \right\rfloor  \leq \tau_s \leq \left\lceil \frac{D_0}{(\mu - \lambda)} \right\rceil.
\]
\end{theorem}
\begin{proof}
The detailed proof of all parts of the theorem can be found in Appendix \ref{app:proofs}. 
\end{proof}

Theorem \ref{thm:Ds_conv} provides us three key insights about $D_s(t)$: i) $D_s(t)$ is always non-increasing; ii) $D_s(t)$ converges to $0$ in finite time; and finally, iii) there are interpretable upper and lower bounds on the number of time steps to convergence. Recall that with $\mu > \lambda$, $(\mu - \lambda)$ represents the excess supply in the zone at each time step after the exogenous demand for that period has been matched and cleared. Therefore, if all riders happened to be \emph{non-strategic} and chose not to walk, the excess supply in each time step would be used to clear the pending demand from $t = 0$ and it would take exactly $\left \lceil \frac{D_0}{\mu-\lambda} \right \rceil$ time steps for all the demand to dissipate. This shows that the upper bound is tight. In reality, however, the surge demand will dissipate faster than this because of the effect of unmatched customers from the surge zone gradually relocating to non-surge zones over time. On the other hand, the lower bound is obtained by noting that $(\mu-\lambda) + D_0 \cdot f(D_0)$ represents the fastest possible rate of clearing demand ($D_0 \cdot f(D_0)$ represents the maximum outflow of strategic riders in any time period which can be achieved at $t = 0$ from the surge to the non-surge zone, and $\mu-\lambda$ the rate that is cleared by matching demand to supply). For $t > 0$, the rate of clearing demand can only be smaller because of the smaller value of $D_s(t)$ and the smaller demand differential across the surge boundary, leading to a smaller total amount of agents moving. Therefore, the surge demand cannot dissipate earlier than $ \left\lfloor \frac{D_0}{(\mu - \lambda) + D_0 \cdot f(D_0)} \right\rfloor$ time steps. This bound is exactly tight in the corner case where the rates of clearing are so high that $(\mu - \lambda) + D_0 \cdot f(D_0) \geq D_0$ and the entire demand is cleared in time step $1$.

Now, we investigate the behavior of $D_{ns}(t)$. We present a complete characterization of the shape properties of $D_{ns}(t)$ in Theorem \ref{thm:Dns_conv}: 
%, followed by an intuitive interpretation of what the theorem means.  

\begin{theorem}\label{thm:Dns_conv}
Let $D_{ns}(t)$ be the unmet demand in the non-surge zone at any time step $t \in \{0, 1, 2,...\}$. Then, there always exists some time step $\tau \geq 0$ such that $D_{ns}(t)$ is non-increasing for all $t \geq \tau$ and increasing for all $t < \tau$. Further, $\tau$ satisfies:
\[
   \tau = \min \left\{t \in \{0, 1, 2,....\}: f\left(D_s(t)-D_{ns}(t) \right)\cdot D_s(t) \leq \mu - \lambda \right\}.
\]
Finally, $D_{ns}(t)$ must converge to $0$ in finite time and if $\tau_n$ indicates the time to convergence, then $\tau_n$ must satisfy: 
\[
    \left\lfloor \frac{d_0}{(\mu - \lambda)} \right\rfloor \leq \tau_n \leq \left\lceil \frac{D_0+d_0}{(\mu - \lambda)} \right\rceil.
\]
\end{theorem}
\begin{proof}
The detailed proof of all parts of the theorem can be found in Appendix \ref{app:proofs}. 
\end{proof}

This theorem provides the following insights: i) if $\tau = 0$, $D_{ns}(t)$ is non-increasing for all $t$ and converges to $0$, and ii) if $\tau > 0$, $D_{ns}(t)$ increases initially to attain a \emph{peak} at $\tau$, followed by which it decreases and converges to $0$. It also provides interpretable upper and lower bounds on the time to convergence. The upper bound corresponds to the extreme case when all riders in the surge zone relocate to the non-surge zone in the first period itself and since the excess supply in each period is $(\mu-\lambda)$ in the non-surge zone, it takes about $ \frac{D_0+d_0}{(\mu - \lambda)} $ time steps to clear a total demand of size $D_0 + d_0$. Similarly, the lower bound corresponds to the other extreme when riders are \emph{non-strategic} and choose not to move, in which case it takes about $\frac{d_0}{(\mu-\lambda)}$ time steps to clear a demand of size $d_0$. Thus, both upper and lower bounds are tight. 

We have already noted that the value of $\tau$ determines the shape of the demand curve $D_{ns}(t)$ in the non-surge zone. We will see shortly that the shape of $D_{ns}(t)$ can be used to characterize the \emph{type of surge}. 
\begin{defn}[Localized Surge]
When the surge demand does not accumulate outside the surge boundary and the unmet non-surge demand $D_{ns}(t)$ is non-increasing for all $t \geq 0$, we say that the surge is \textbf{localized}. Formally, $\tau = 0$ corresponds to the \textbf{localized surge} scenario.  
\end{defn}

\begin{defn}[Spill-over Surge]
When the surge demand accumulates outside the surge zone causing the unmet non-surge demand $D_{ns}(t)$ to increase (even if for a short period of time), we say that the surge has \textbf{spilled over}. Formally, $\tau > 0$ corresponds to the \textbf{spill-over surge} scenario. 
\end{defn}

$\tau > 0$ represents a different scenario where up to time step $\tau$, $D_{ns}$ is increasing because the mass of riders relocating across the surge boundary is large enough that it cannot be cleared instantaneously by the excess supply $\mu - \lambda$, leading to accumulating demand. Since the high demand inside the surge zone appears to \emph{spill-over} and affect demand outside the surge boundary, we call this type of surge a \emph{spill-over}. In contrast, we call the case of $\tau = 0$ a \emph{localized surge}, where the inflow of riders from the surge into the non-surge zone is less than the excess supply, and $D_{ns}$ is decreasing.  %because the surge demand does not lead to accumulation of demand outside the surge zone. As a result, the shape of $D_{ns}(t)$ remains unaffected (as if there was no surge in the surge zone) and it decreases steadily to $0$. 

In the case of a spill-over surge, as $D_{ns}(t)$ increases, the demand differential $D_s(t) - D_{ns}(t)$ decreases (recall, $D_s(t)$ is steadily decreasing) and so does the incentive for riders to relocate. As a result, beyond $\tau$, $D_{ns}(t)$ starts to get cleared by the excess supply $\mu - \lambda$. This leads to the rapid decrease and eventual convergence of $D_{ns}(t)$. An interesting highlight of Theorem \ref{thm:Dns_conv} is that $D_{ns}(t)$ has a relatively simple shape in that it is \emph{single peaked}.%, i.e., it does not exhibit any non-monotone or oscillating behavior beyond $\tau$. %Reasoning about the behavior of $D_{ns}(t)$ beyond $\tau$ is non-trivial.   

We now provide a necessary and sufficient condition on the initial demands $D_0$ and $d_0$ that guarantees $\tau = 0$ or equivalently, the absence of the \emph{spill-over effect}. This formalizes the intuition expressed above about the relationship between rider relocation and excess supply: 

\begin{claim}\label{clm:tau_eq0}
There is no spill-over, i.e., $\tau = 0$ if and only if the initial demands $D_0$ and $d_0$ satisfy $D_0 \cdot f\left(D_0 - d_0\right) \leq \mu - \lambda$.
\end{claim}
\begin{proof}
The proof follows directly from the definition of $\tau$ in Theorem \ref{thm:Dns_conv}.  
\end{proof}

\section{Numerical Experiments}\label{sec:experiments}
In this section, we will start by visualizing the main theoretical insights that we have developed in Section~\ref{sec:dynamics} through numerical experiments---specifically the shapes of the demand curves under different surge types, whether the convergence times satisfy our theoretical bounds and how tight the bounds are. We then use simulations to extend our analysis to more general settings with stochastic arrival rates, explicit platform prices, strategic drivers reacting to prices, and a game-theoretic model of rider decision-making. 

\subsection{Deterministic Setting}\label{sec:experiments_deterministic}
First, we present numerical experiments for the simple deterministic setting captured in our model in Section \ref{sec:model} where the exogenous demand $\lambda$ and the supply of drivers $\mu$ in each zone is deterministic and fixed for all time periods. The main objectives for this segment are as follows: i) observing how the unmet demands in the two zones $D_s(t)$ and $D_{ns}(t)$ evolve over time $t$ in \emph{localized} and \emph{spill-over} surge scenarios, ii) identifying and highlighting interesting phenomenon like \emph{surge inversions} that can arise in special cases, and iii) investigating how fast the demands converge compared to the bounds predicted by theory and how the relative size of parameters $\mu$ and $\lambda$ affect this convergence. 

For these experiments, we choose a simple functional form of $f(\cdot)$ which captures the fraction of waiting customers in the surge zone who choose to walk across the surge boundary in each time period. We use $f(x) = \max \left(0,\frac{kx}{\mu} \right)$ where $k$ is a scalar multiplier. Therefore the rate of movement from the surge to the non-surge zone is given by 
\begin{align}\label{eq:theory_movement}
\max \left(0,\frac{k(D_s - D_{ns})}{\mu}\right).
\end{align}

The value of $k$ is chosen appropriately to ensure that $f(\cdot) \in [0, 1]$. A higher value of $k$ indicates that a large fraction of customers choose to move. Recall that the input $x$ to $f(\cdot)$ is the demand difference between the surge and non-surge zones (Assumption \ref{as:strat_riders}).  We additionally use the supply $\mu$ as a normalizing factor, referring back to our pricing rule example in Example \ref{ex:1}, and noting that intuitively, the higher the supply, the lesser the waiting times and the likelihood of a surge, and the less likely agents are to move. Finally, our justification for using a simple linear form for $f(\cdot)$ is that it is easy to implement and interpret, but we will show later in Section \ref{sec:experiments_agentbased} that despite being simple, our linear model of $f(\cdot)$ has very good explaining power and can accurately capture the effect of moving riders due to high prices in the surge zone.

\begin{figure}[!ht]
      \centering
      \includegraphics[width=0.6\textwidth]{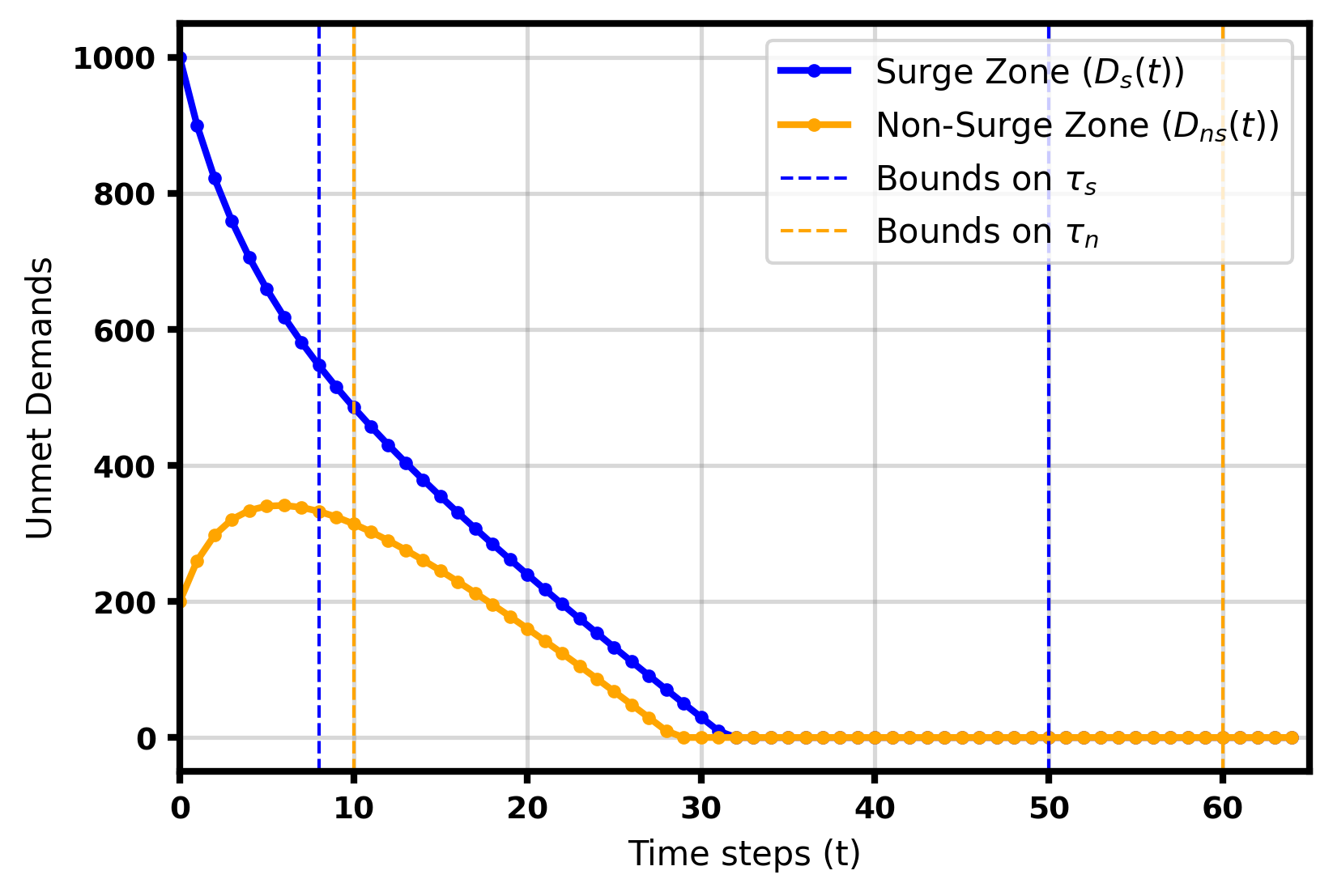}
      \caption{Spill-over surge: Demand curves over time for $D_0 = 1000$, $d_0 = 200$, $\lambda = 30$ and $\mu = 50$. We set $k = 0.005$ and simulate over $T = 65$ time steps. $D_s(t)$ converges in $\tau_s = 33$ time steps while $D_{ns}(t)$ converges in $\tau_n = 30$ time steps. Observe that they lie within our theoretically computed bounds in Theorems \ref{thm:Ds_conv} and \ref{thm:Dns_conv} (indicated by dotted lines).}
    \label{fig:spill_surge}
\end{figure}
In Figure \ref{fig:spill_surge}, we demonstrate a spill-over surge scenario. Note that for the chosen values of the parameters $D_0$, $\lambda$, $\mu$, we have $D_0 \cdot f\left( D_0-d_0\right) > \mu - \lambda$ which we showed in Claim \ref{clm:tau_eq0} to be a sufficient condition for surge to spill-over. This results in $D_{ns}(t)$ having a single peak at some $t > 0$ and then decreasing to $0$. We also observe that $D_s(t)$ decreases monotonically to $0$, as expected.  
\begin{figure}[!ht]
      \centering
      \includegraphics[width=0.6\textwidth]{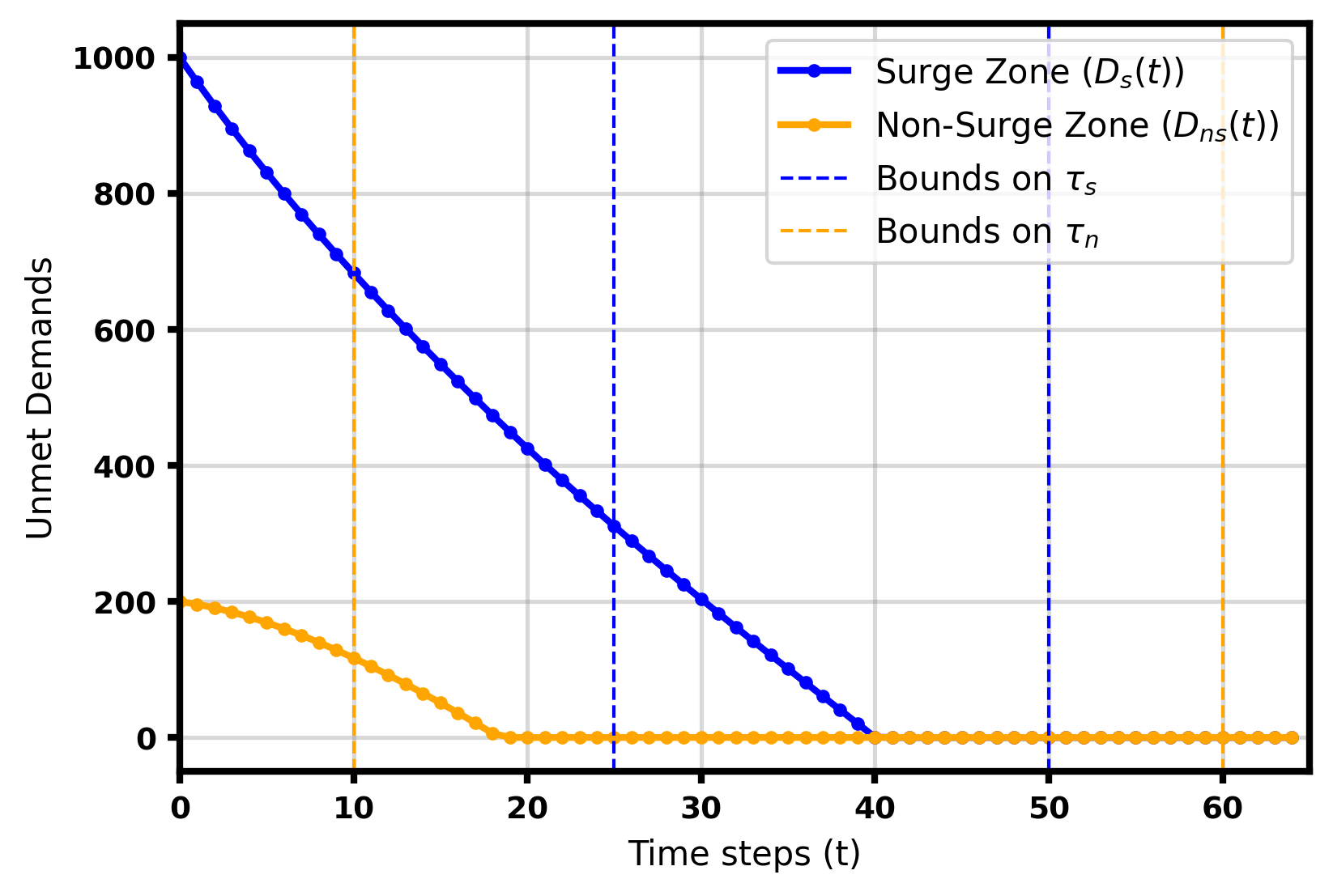}
      \caption{Localized surge: Demand curves over time for $D_0 = 1000$, $d_0 = 200$, $\lambda = 30$ and $\mu = 50$. We set $k = 0.001$ and simulate over $T = 65$ time steps. $D_s(t)$ converges in $\tau_s = 41$ time steps while $D_{ns}(t)$ converges in $\tau_n = 20$ time steps (within theoretically computed bounds). Note the the distinctive difference in shape of $D_{ns}(t)$ from Figure \ref{fig:spill_surge}. }
      \label{fig:local_surge}
\end{figure}

Now, Figure \ref{fig:local_surge} demonstrates a scenario with a low value of $k$---here the mass of riders relocating in each time period is always smaller than $\mu - \lambda$, so it can be matched instantaneously by the excess supply in the non-surge zone. This leads to a localized surge scenario ($\tau = 0$) where the unmet demands in both zones are found to decrease monotonically as predicted by theory. This also causes $D_{ns}(t)$ to converge much faster than $D_s(t)$.   

\paragraph{Surge Inversion:} We highlight the existence of cases where \emph{the non-surge demand $D_{ns}(t)$ exceeds the surge demand $D_s(t)$}. We call this phenomenon a \emph{surge inversion}. An example of surge inversion is presented below in Figure \ref{fig:surge_inv}. We have already seen that when the surge does not remain localized, $D_{ns}(t)$ increases initially and attains a peak before decreasing monotonically.  
\begin{figure}[!ht]
      \centering
      \includegraphics[width=0.6\textwidth]{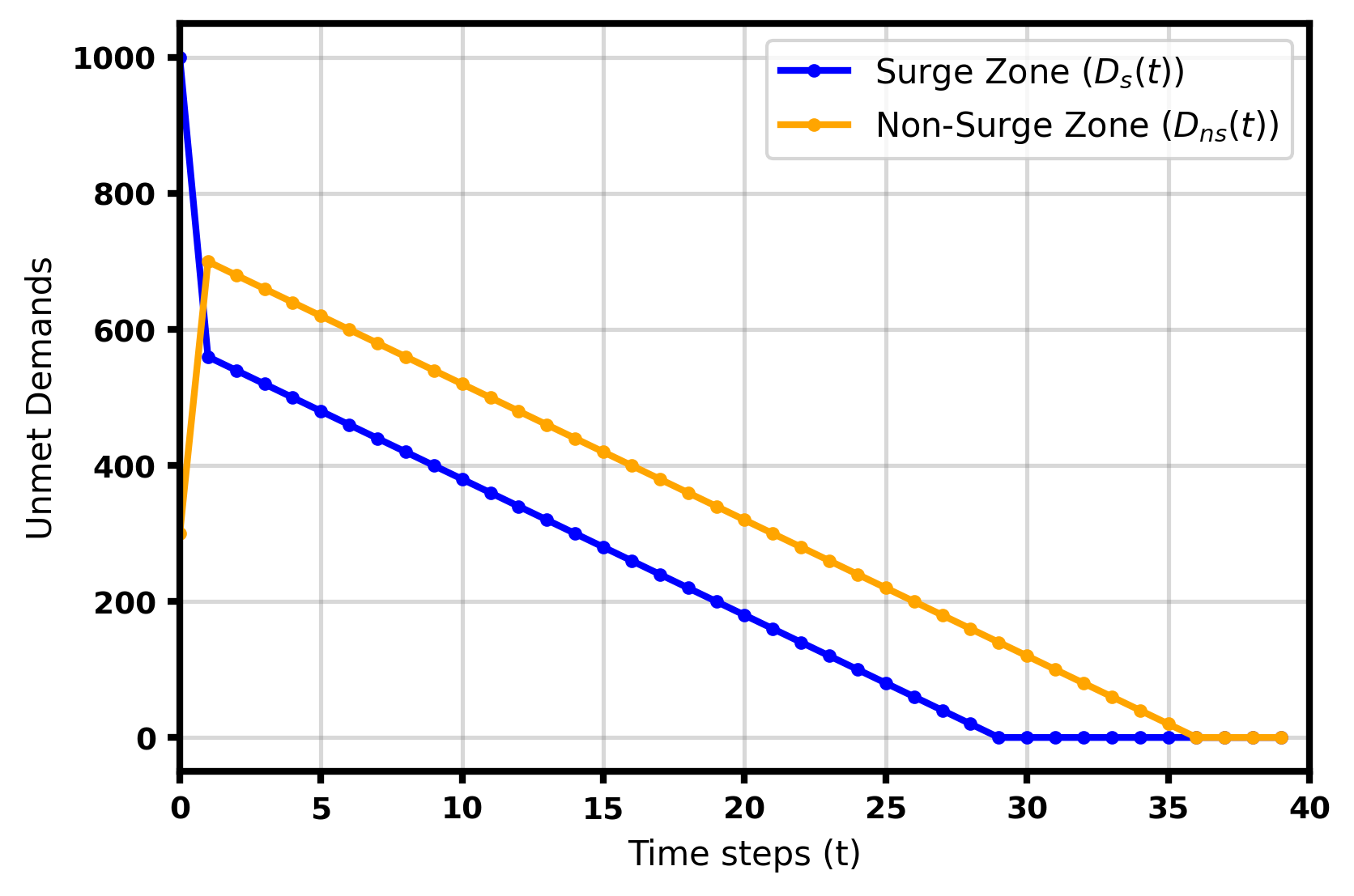}
      \caption{Surge Inversion: Demand curves over time for $D_0 = 1000$, $d_0 = 300$ $\lambda = 30$, $\mu = 50$. We set $k = 0.05$. Observe that for this set of parameters, $D_{ns}(t)$ exceeds $D_s(t)$ at $t = 1$.}
      \label{fig:surge_inv}
\end{figure}
However, in some cases (at high values of $k$), the mass of riders who choose to walk over in the first few time periods is so significant that the surge demand $D_s(t)$ dips below the peak of $D_{ns}(t)$ leading to a \emph{surge inversion}. As soon as inversion happens, the status quo is maintained for all future time steps because as per Assumption \ref{as:strat_riders}, riders will never choose to walk back inside the surge zone again even if the demand outside is higher\footnote{When riders are allowed to walk back, we expect to see cycles where the zone with the highest demand and highest price switches over time; we do not aim to characterize such cycles in this paper.}. Henceforth, both zones will clear demand at a fixed rate of $(\mu - \lambda)$ until they converge to $0$, implying that $D_s(t)$ converges faster than $D_{ns}(t)$ during \emph{surge inversion}. It is important to note that surge inversion does not violate any of our theoretical bounds on $D_s(t)$ or $D_{ns}(t)$.   

\paragraph{Variation of results over the parameter space:} Finally, we conclude the discussion on the deterministic setting with an analysis of how changing the value of the different parameters like $D_0$, $k$, $\mu$ and $\lambda$ may affect the results: we are primarily interested in seeing i) how the system transitions between the two types of surge as a function of the parameters, and ii) how different parameters affect the time to convergence of the demands relative to each other. 

In Figure \ref{fig:mulambdacomp_surgetype}, we fix the initial demands $D_0$, $d_0$ and the parameter $k$ associated with $f(\cdot)$ and only vary the exogenous demand $\lambda$ and supply rate $\mu$. Firstly, as the magnitude of the excess supply ($\mu-\lambda$) increases, the gap between the convergence times $\tau_s$ and $\tau_n$ increases --- $D_{ns}(t)$ is found to converge much faster than $D_s(t)$. The excess supply also determines the type of surge---with a smaller value of ($\mu-\lambda$), \emph{spill-over} surges are more likely as long as the supply rate $\mu$ is not too large; this is in accordance with Claim \ref{clm:tau_eq0}. The size of the supply $\mu$ is also important independently of $\lambda$, as it plays a role in the rate at which agents move, given in Equation~\ref{eq:theory_movement}. Because the initial demands $D_0$ and $d_0$ are fixed, a smaller value of $\mu$ indicates that the demand gap across the surge boundary relative to the supply rate is higher, which increases the incentive for riders to walk. This causes the $D_s(t)$ and $D_{ns}(t)$ curves to become steeper and approach each other, as the demand spontaneously redistributes.        

\begin{figure}[!ht]
    \centering
    \raisebox{35pt}{\parbox[b]{.05\textwidth}{}}%
    \subfloat[][$\mu$ = 30, $\lambda$ = 5]{\includegraphics[width=.3\textwidth]{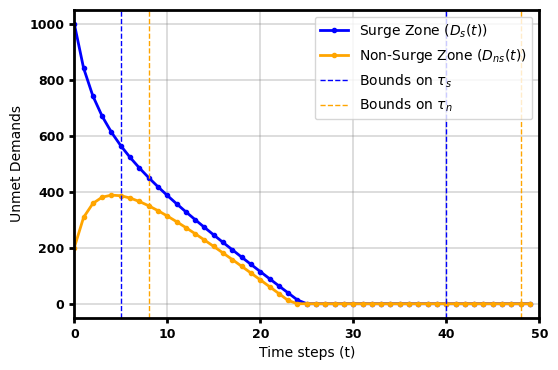}}\hfill
    \subfloat[][$\mu$ = 30, $\lambda$ = 15]{\includegraphics[width=.3\textwidth]{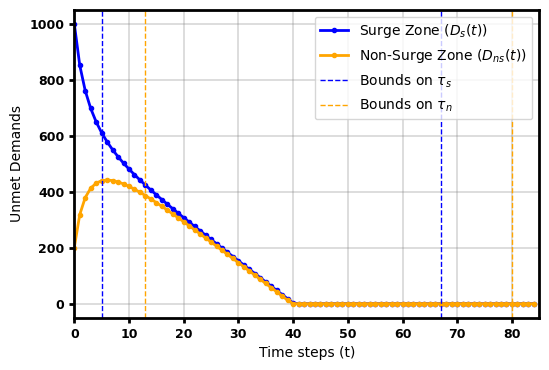}}\hfill
    \subfloat[][$\mu$ = 30, $\lambda$ = 25]{\includegraphics[width=.3\textwidth]{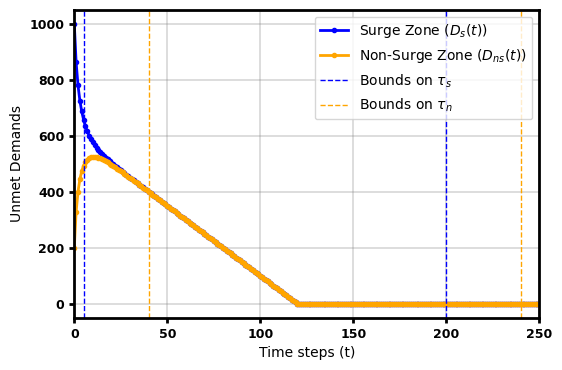}}\par
    
    \raisebox{35pt}{\parbox[b]{.05\textwidth}{}}%
    \subfloat[][$\mu$ = 85, $\lambda$ = 10]{\includegraphics[width=.3\textwidth]{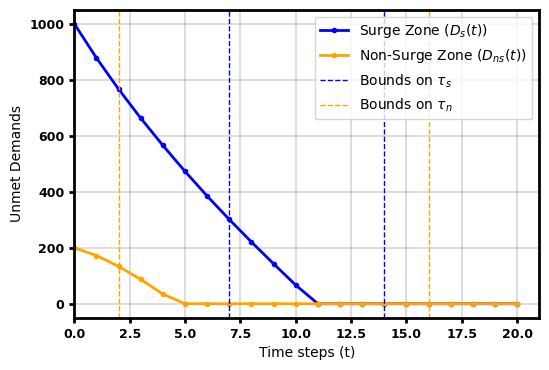}}\hfill
    \subfloat[][$\mu$ = 85, $\lambda$ = 50]{\includegraphics[width=.3\textwidth]{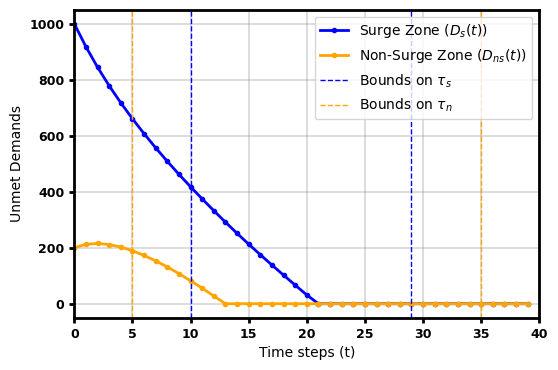}}\hfill
    \subfloat[][$\mu$ = 85, $\lambda$ = 70]{\includegraphics[width=.3\textwidth]{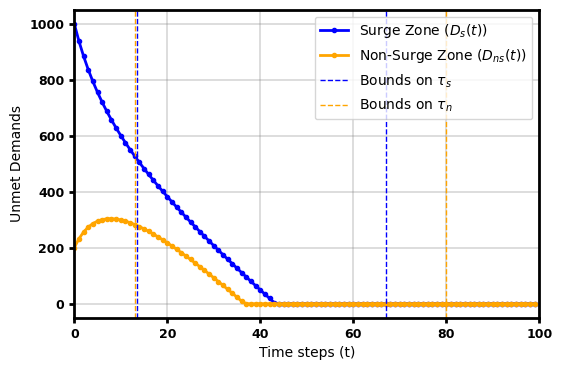}}\par
    
    \raisebox{35pt}{\parbox[b]{.05\textwidth}{}}%
    \subfloat[][ $\lambda$ = 15, $\mu$ = 20]{\includegraphics[width=.3\textwidth]{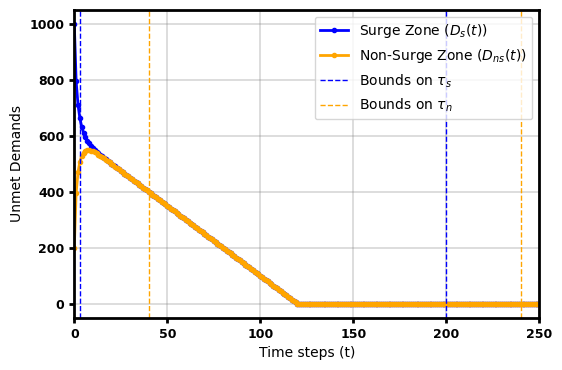}}\hfill
    \subfloat[][$\lambda$ = 15, $\mu$ = 40]{\includegraphics[width=.3\textwidth]{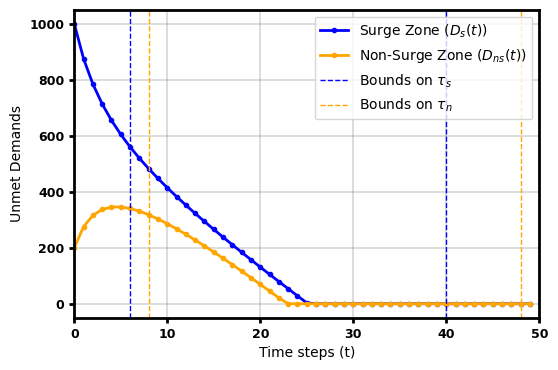}}\hfill
    \subfloat[][$\lambda$ = 15, $\mu$ = 60]{\includegraphics[width=.3\textwidth]{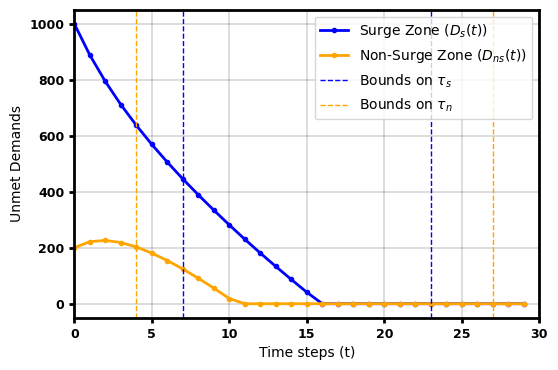}}\par
    
    \raisebox{35pt}{\parbox[b]{.05\textwidth}{}}%
    \subfloat[][$\lambda$ = 60, $\mu$ = 70]{\includegraphics[width=.3\textwidth]{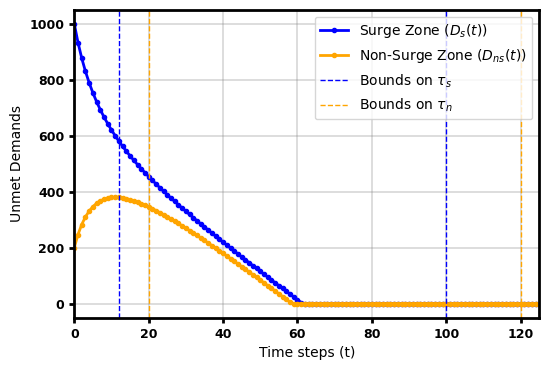}}\hfill
    \subfloat[][$\lambda$ = 60, $\mu$ = 90]{\includegraphics[width=.3\textwidth]{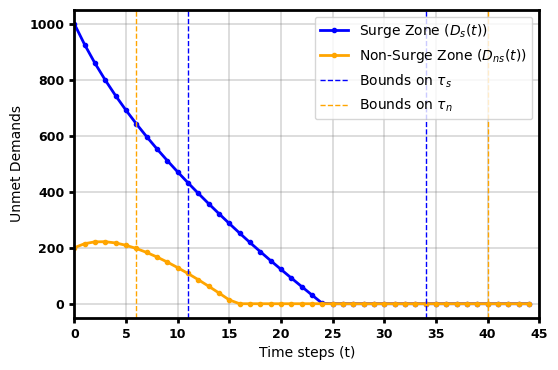}}\hfill
    \subfloat[][$\lambda$ = 60, $\mu$ = 110]{\includegraphics[width=.3\textwidth]{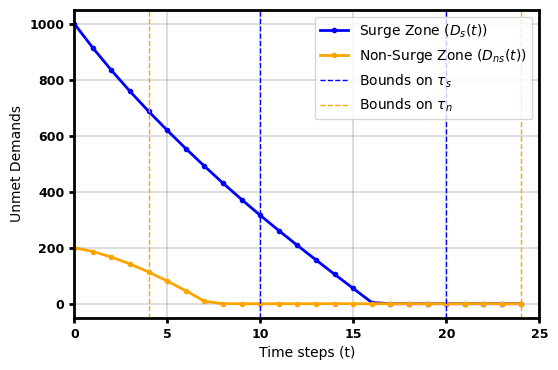}}\par

    \caption{Illustration of surge dynamics in the deterministic setting with variation in parameters $\mu$ and $\lambda$. For all sets of results here, the following parameter values are fixed: $D_0 = 1000$, $d_0 = 200$, and $k = 0.005$. }
    \label{fig:mulambdacomp_surgetype}
\end{figure} 

In Figure \ref{fig:D0_k_vary}, we explore how the demand curves evolve when i) the volume of the initial surge demand $D_0$ changes (all else fixed), and ii) the parameter $k$ in $f(\cdot)$ changes (all else fixed). From sub-figures (a)-(b)-(c), it is clear that as $D_0$ increases, the surge type gradually transitions from the localized to the spill-over type which is highly intuitive. We make a similar observation in sub-figures (d)-(e)-(f) where $D_0$ is fixed, but $k$ increases gradually. A higher $k$ indicates that a higher fraction of riders choose to walk in each time step. Again, we see a gradual evolution in the shape of $D_{ns}(t)$ as it transitions from being monotonically decreasing (localized surge) to increasing first and then decreasing (spill-over surge).   

Across both Figures \ref{fig:mulambdacomp_surgetype} and \ref{fig:D0_k_vary}, it is important to note how close the convergence times $\tau_s$ and $\tau_n$ are to the bounds. In localized surge settings, $\tau_s$ is very close to its upper bound and $\tau_n$ is very close to its lower bound. These observations can be  explained by noting that in localized surge, the mass of riders walking across the surge boundary is not too high---as a result, the surge and non-surge demands only change mildly (surge demand decreases, non-surge demand increases) from the scenario where riders are non-strategic and do not walk. However, as we move to spill-over surge setting, our bounds become looser which is again along expected lines (in the case of surge demand, the upper bound completely ignores the effect of walking while the lower bound overestimates it by using the rate of walking in the first time step which is the highest).

\begin{figure}[!ht]
  \centering
  \raisebox{20pt}{\parbox[b]{.11\textwidth}{}}%
  \subfloat[][$D_0$ = 300, $k = 0.005$ ]{\includegraphics[width=.3\textwidth]{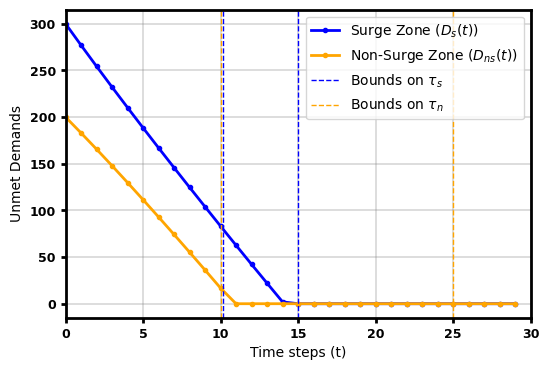}}\hfill
  \subfloat[][$D_0$ = 500, $k = 0.005$]{\includegraphics[width=.3\textwidth]{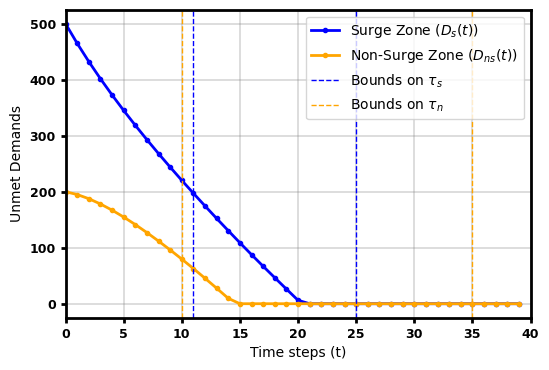}}\hfill
  \subfloat[][$D_0$ = 1500, $k = 0.005$]{\includegraphics[width=.3\textwidth]{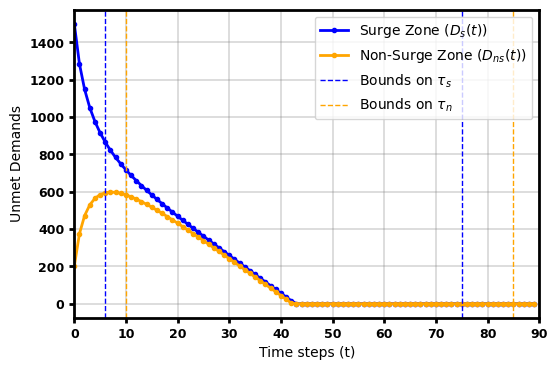}}\par
  
  \raisebox{20pt}{\parbox[b]{.11\textwidth}{}}%
  \subfloat[][$k = 0.0005$, $D_0$ = 1000]{\includegraphics[width=.3\textwidth]{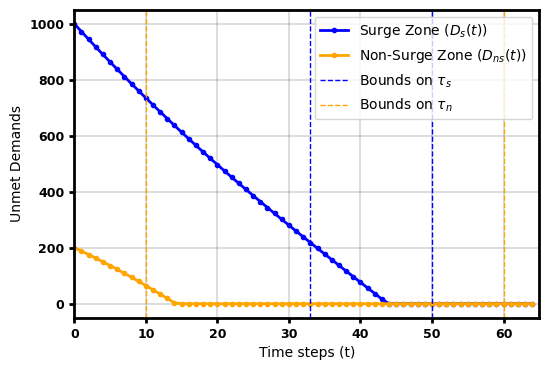}}\hfill
  \subfloat[][$k = 0.001$, $D_0$ = 1000]{\includegraphics[width=.3\textwidth]{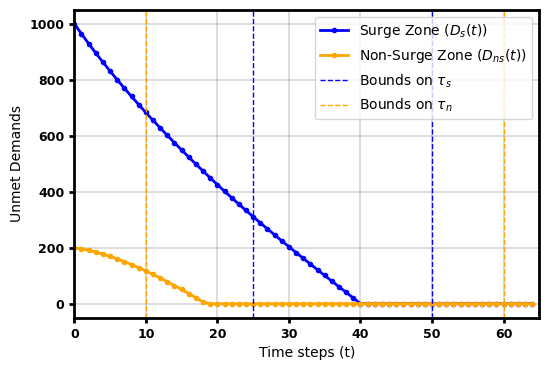}}\hfill
  \subfloat[][$k = 0.01$, $D_0$ = 1000]{\includegraphics[width=.3\textwidth]{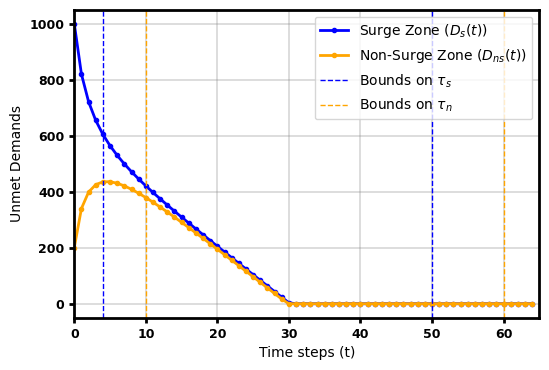}}\par

  \caption{Demonstration of how the shapes of $D_s(t)$ and $D_{ns}(t)$ change as we traverse the parameter space in $k$ and $D_{0}$. All other parameters remain fixed: $d_0 = 200$, $\lambda = 30$ and $\mu = 50$.}
  \label{fig:D0_k_vary}
\end{figure}
% \jz{Figure~\ref{fig:var} does not bring any value to the paper. We can't even match which blue curve corresponds to which yellow curve and when/how often Dns can overtake Ds. I would delete. Instead, I would have several plots/a wrapfig with different values of the parameters.}\ds{The purpose of this figure is to show how the shapes of the curves transform as we go from the localized to the spillover case with changing values of $D_0$ and $k$. Anyway, I will delete this.}\jz{but it is really really hard to see here from this plot. In this case we want several figures or to make it more sparse, but right now it is really hard to read/does not bring much value}\ds{Cool, I will try to redo the plots in a better way. If it doesn't work out, we can just drop it.}

\subsection{Extension 1: Stochastic Demands \& Supplies}\label{sec:experiments_stochastic}
One immediate extension is to consider stochastic exogenous demands and supplies for both zones at all time steps. Instead of having a fixed exogenous demand $\lambda$ and a fixed supply $\mu$ at each time step, we now assume that at any time $t \in \{0, 1, 2,....T\}$, the exogenous demand is drawn from a distribution $F(\cdot)$ with mean number of arrivals per unit time $\lambda$. Similarly, the supply is drawn from another distribution $G(\cdot)$ with mean number of arriving vehicles per unit time $\mu$. We choose a Poisson distribution for $F$ and $G$, a standard process for modeling stochastic arrival \cite{poisson}.

\begin{figure}[!ht]
      \centering
      \includegraphics[width=0.6\textwidth]{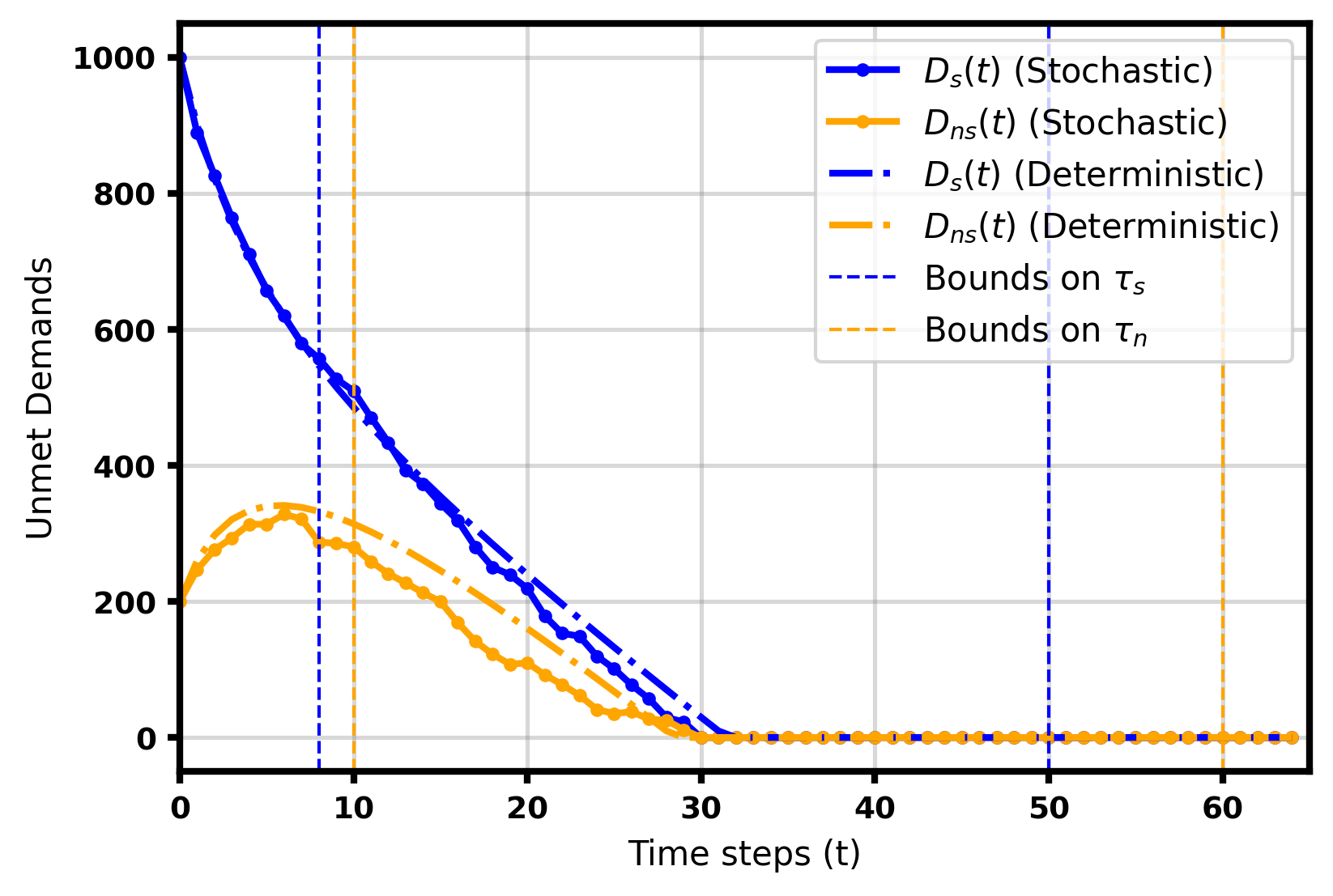}
      \caption{Spill-over surge under stochastic demands and supplies: Demand curves over time for $D_0 = 1000$, $d_0 = 200$, $\lambda = 30$ and $\mu = 50$. We set $k = 0.005$ and simulate over $T = 65$ time steps. For one realization of supplies and demands, $D_s(t)$ converges in $\tau_s = 31$ time steps while $D_{ns}(t)$ converges in $\tau_n = 29$ time steps (compared to $\tau_s = 33$ and $\tau_n = 30$ in the deterministic case). }
      \label{fig:stoch_spill_surge}
\end{figure}

\begin{figure}[!ht]
      \centering
      \includegraphics[width=0.6\textwidth]{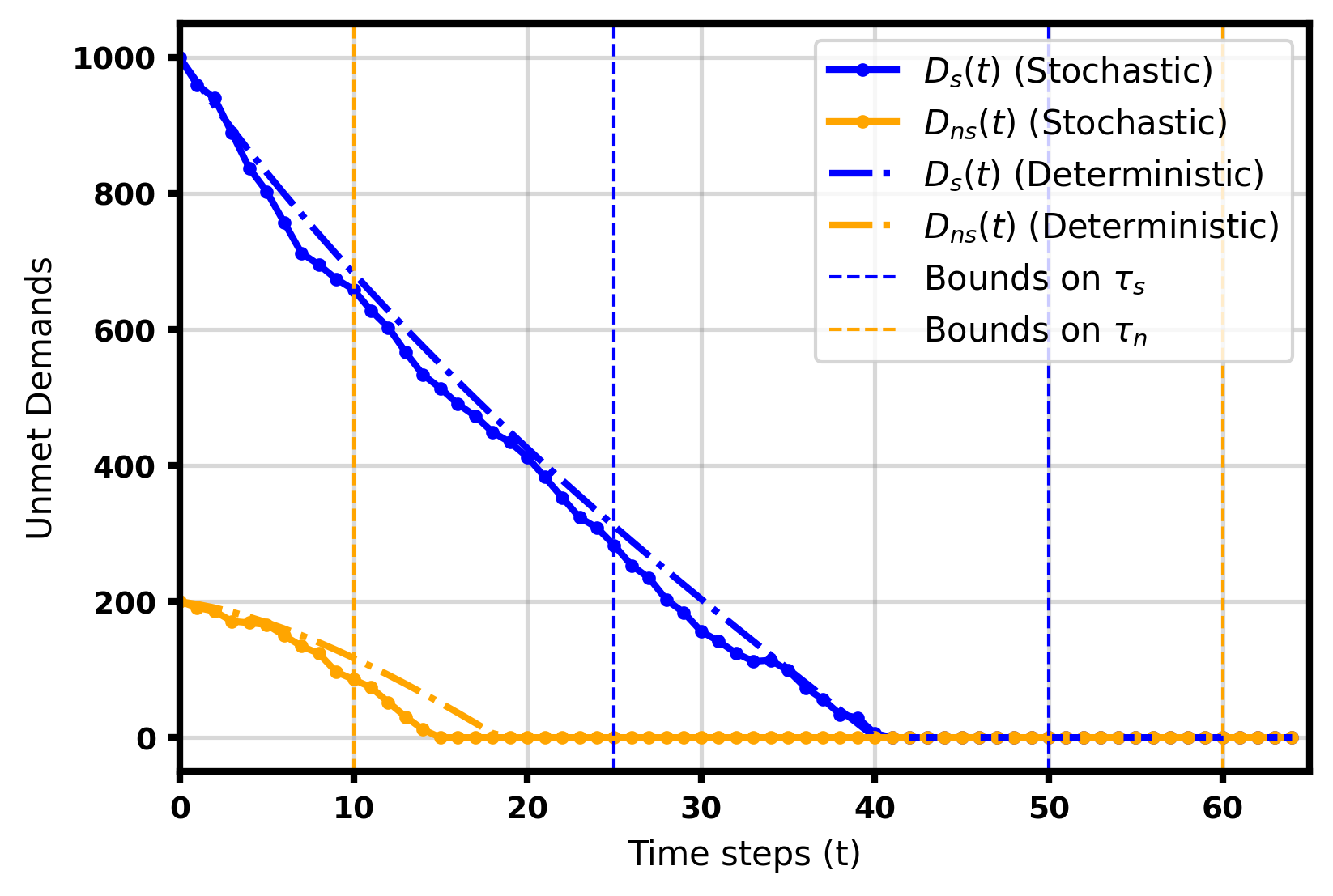}
      \caption{Localized surge under stochastic demands and supplies: Demand curves over time for $D_0 = 1000$, $d_0 = 200$, $\lambda = 30$ and $\mu = 50$. We set $k = 0.001$ and simulate over $T = 65$ time steps. For one realization of supplies and demands, $D_s(t)$ converges in $\tau_s = 42$ time steps while $D_{ns}(t)$ converges in $\tau_n = 20$ time steps (compared to $\tau_s = 41$ and $\tau_n = 20$ in the deterministic case).}
      \label{fig:stoch_local_surge}
\end{figure}
We replicate our experiments with stochastic demands and supplies for the same parameter combinations as in the deterministic case (Figures \ref{fig:stoch_spill_surge} and \ref{fig:stoch_local_surge}). We can see that the trends broadly generalize from the deterministic to the stochastic case, in terms of convergence and shape properties of $D_s(t)$ and $D_{ns}(t)$.

\subsection{Extension 2: Agent-based, Game-theoretic simulation of ridesharing system with real-time market matching and passenger movement}\label{sec:experiments_agentbased}

Our theoretical results, as well as our experimental results of Section~\ref{sec:experiments_stochastic}, rely on several strong assumptions. The first one is that, given the unmet demand $D_s(t)$ in the surge zone and the unmet demand $D_{ns}(t)$ in the outside zone, the number of riders moving every day is a function of $D_s - D_{ns}$. The second one is that drivers are \emph{not strategic} and arrive in each zone at a rate that is independent of demands and prices, while drivers may chase surges in real life. In this section, we provide experiments showing that our theoretical highlights are robust to our theoretical modeling choice, highlighting how our model of Section~\ref{sec:model} provides a reasonable, first-order approximation of the main effect we are modeling. Our experiments rely on an \emph{agent-based}, \emph{game-theoretic} model of how riders decide to walk out of the surge zone and how drivers decide between the surge and non-surge zones. 

\subsubsection{Game-theoretic Model.}\label{sec:agent_model}
Below, we describe our game-theoretic, game-theoretic model of interactions between riders, drivers, and the platform. 

\paragraph{Driver Model:} We consider a finite number of \emph{strategic} drivers. Drivers arrive in each time step at a rate of $2 \mu$, according to a \emph{Poisson} distribution\footnote{In this sense, we do not model drivers who strategically decide when to come online.}. However, drivers can strategically choose \emph{which of the surge and non-surge zones} to serve. Given price multipliers $P_s$ and $P_{ns}$ in the surge and non-surge zones respectively, each driver decides which zone to go to according to a \emph{multinomial logit model}, a seminal and practically well-motivated model of decision-making~\cite{hausman1984specification}: i.e., a driver goes to the surge zone with probability
\[
\gamma_s = \frac{\exp(\tau P_s)}{\exp(\tau P_s) + \exp(\tau P_{ns})},
\]
and to the non-surge zone with probability 
\[
\gamma_{ns} = \frac{\exp(\tau P_{ns})}{\exp(\tau P_s) + \exp(\tau P_{ns})},
\]
for some positive parameter $\tau > 0$.

\paragraph{Rider Model:} We consider a finite number of riders, arriving at a rate $\lambda$ in each zone, according to a \emph{Poisson} distribution. There is also an initial number of riders $D_0$ in the surge zone and $d_0$ in the non-surge zone. 

Each rider has a \emph{willingness to pay}, denoted by $v$, which determines whether they choose to accept a ride, given a price $p$. If $v \geq p$, then the rider can be matched with a ride at the current price, otherwise they continue to wait. Each rider is also associated with a \emph{cost} $c$ for moving from the surge zone to the non-surge zone. The costs $c$ are drawn i.i.d. from a distribution $\mathcal{D}$ with cumulative distribution function $F(.)$. In this case, if the prices set by the platform are $P_{ns}$ and $P_s$, a user with cost $c$ in the surge zone walks to the non-surge zone if and only if
\[
c \leq P_{s} - P_{ns}.
\]
In particular, the probability that a given user decides to walk to the non-surge zone is then given by $F(P_{s} - P_{ns})$. In our experiments, we choose $\mathcal{D}$ to be a \emph{symmetric} truncated Gaussian distribution (left-truncated at $0$) for which the target mean $\mathcal{D}_{mean}$ and standard deviation $\mathcal{D}_{std}$ are specified. $\mathcal{D}_{mean}$ and $\mathcal{D}_{std}$ directly influence what fraction of the current surge demand $D_s(t)$ chooses to relocate across the surge boundary.

\paragraph{Platform Prices:} We consider a platform using a uniform price multiplier $P_{s}$ in the surge zone and $P_{ns}$ in the non-surge zone. We are in particular interested in studying the price gap $\Delta P = P_s - P_{ns}$. Note that $\Delta P$ represents an additive surge premium over the non-surge price;~\citet{garg2022driver} have shown that additive surge pricing has desirable properties, such as being incentive-compatible for drivers. Additionally, ~\citet{ma2022spatio} demonstrate that spatial mispricing (i.e., prices in adjacent areas being vastly different) during price surges incentivizes drivers to ``cherry-pick" rides and is a leading cause of market failure. Therefore, guaranteeing low $\Delta P$ across adjacent zones is highly desirable, and we will show in this section how rider strategic behavior contributes to reducing this gap $\Delta P$.

We derive $\Delta P$ following the approach from the previous work of~\citet{yu2023price}, the platform set prices to satisfy a simple equilibrium condition: namely, the wait times in both zones should be roughly the same. Here, we consider a slight simplification of this condition, where we use \emph{the ratio of demand to supply} as a proxy for wait times.\footnote{If there are $D$ riders total to clear, $S$ drivers come at every time step, and $D \geq S$, then the average wait time is $D/S - 1$.} Letting $R_s(t)$ and $R_{ns}(t)$ be the average number of drivers going to the surge and non-surge zones (respectively) at time $t$, the platform's goal is to set equilibrium prices to induce
\[
\frac{D_s(t)}{R_s(t)} = \frac{D_{ns}(t)}{R_{ns}(t)}.
\]
Remember that $R_s(t)$ and $R_{ns}(t)$ are a function of the price multipliers, following a multinomial-logit-based decision process; given an arrival rate of $2 \mu$, we have that in expectation, 
\[
R_s(t) = 2 \mu  \cdot \frac{\exp(\tau P_{s})}{\exp(\tau P_s) + \exp(\tau P_{ns})},~~R_{ns}(t) = 2 \mu  \cdot \frac{\exp(\tau P_{ns})}{\exp(\tau P_s) + \exp(\tau P_{ns})}.
\]
The equilibrium condition for the price multipliers then becomes 
\[
\exp(\tau P_{ns}) \cdot \frac{D_s(t)}{D_{ns}(t)} = \exp(\tau P_{s}),
\]
yielding the following closed-form expression for the \emph{price gap} across the surge boundary:
\begin{align}
\Delta P = P_s - P_{ns} = \frac{1}{\tau} \log \left(\frac{D_s(t)}{D_{ns}(t)} \right).
\end{align}

When the total unmet demand $D_s(t) + D_{ns}(t)$ falls below the total supply, we note that the \emph{surge is cleared}, and in this case $\Delta P = 0$ (both zones return to normal pricing). 

Finally, platforms usually also need to impose a cap $M$ on the surge price either due to regulatory interventions ~\cite{Hawai} or to make sure that prices do not grow unbounded during emergencies ~\cite{Uber}. We incorporate the same into our pricing model.

\subsubsection{Experimental results.}
In this segment, we provide simulation results using the agent model described above where we capture how the unmet demands $D_s(t)$ and $D_{ns}(t)$ evolve over time. We compare these results directly against: 
\begin{itemize}
    \item Results from our theoretical model (Theory); and 
    \item A baseline where riders are \emph{not strategic}, called Non-Strategic Benchmark (NSB). %The benchmark case is straightforward: we observe how $D_s(t)$ and $D_{ns}(t)$ dissipate in the absence of the \emph{strategic riders} i.e., when riders are non-strategic and continue waiting at their zone of origin until they are matched with a ride. It uses the same agent based model introduced earlier in this section, the only difference being that all riders have infinitely large moving costs.
    This baseline uses the game-theoretic model described in Section~\ref{sec:agent_model}, but with the caveat that riders \emph{never} move from the surge to the non-surge zone. Equivalently, riders have a large or infinite cost $c$ for moving across zones.
\end{itemize}
Henceforth, for brevity, we will refer to the strategic agent model as \textsf{(SA)}, the theoretical model as \textsf{(Theory}) and the non-strategic agent benchmark model as \textsf{(NSB)}.

\paragraph{Comparison to \textsf{Theory} model: } We explore the explaining power of our \textsf{theory} model with respect to the strategic agent simulation model (\textsf{SA}). In our \textsf{theory} model, the effect of agents relocating across the surge boundary is modelled using a function $f(\cdot)$ which depends on the demand differential between the surge and the non-surge zones. Recall that we have assumed a simple linear form of $f(\cdot)$ given by $f(x) = \frac{kx}{\mu}$. In Figure \ref{fig:SAvsTheory_diffk}, we present the demand dynamics of our \textsf{theory} model for the following set of parameters ($D_0 = 2000$, $d_0 = 250$, $\lambda = 30$, $\mu = 50$, $\mathcal{D}_{mean} = 15$, $\mathcal{D}_{std} = 8$) for different values of $k$. Observe that $k$ can be tuned to emulate the \textsf{SA} model very closely, which we show in Figure \ref{fig:SAvsTheory_optk}.
\begin{figure}[!ht]
    \centering
    \raisebox{35pt}{\parbox[b]{.05\textwidth}{}}%
    \subfloat[][$k = 0.0005$]{\includegraphics[width=.45\textwidth]{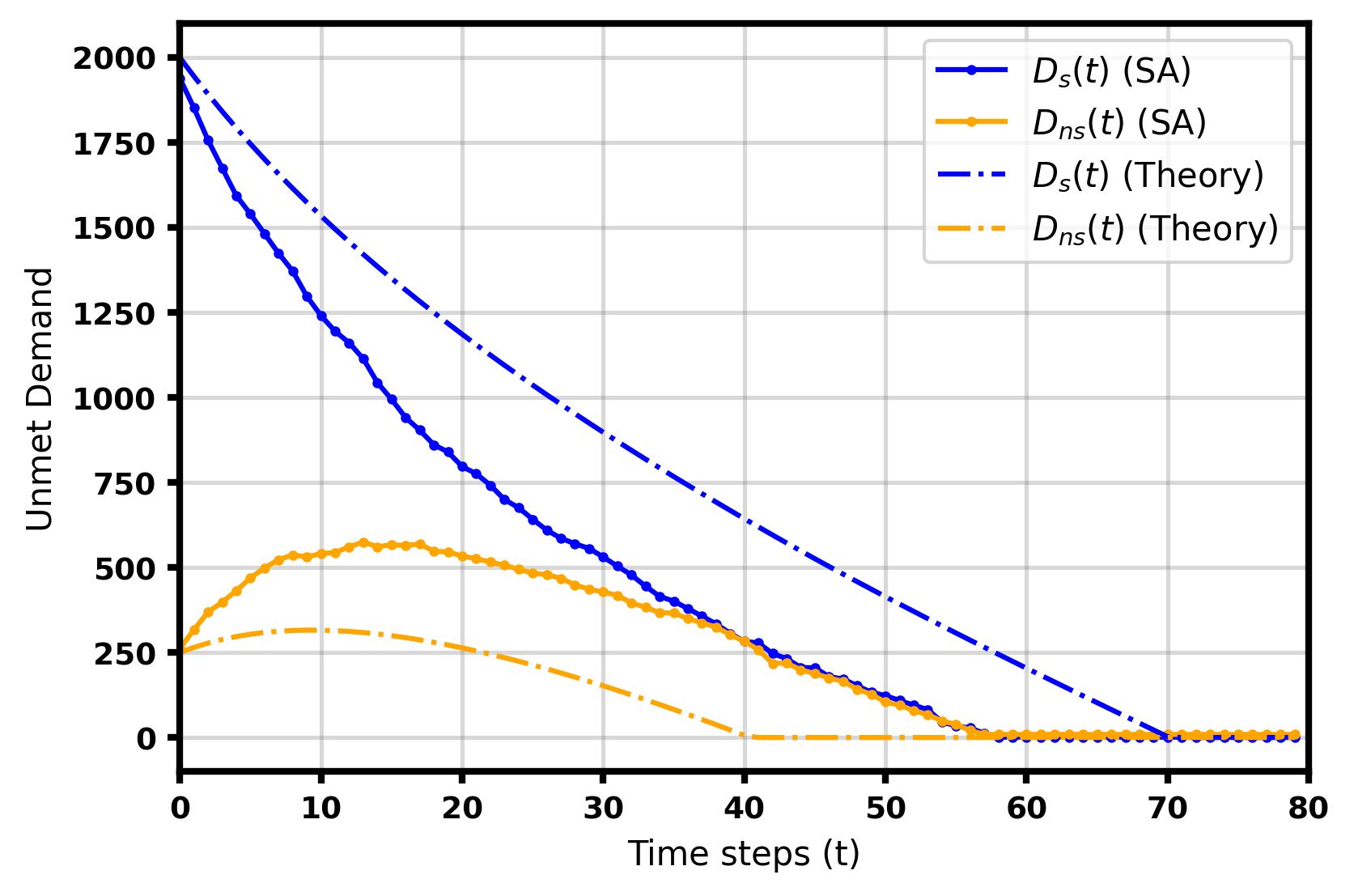}}\hfill
    \subfloat[][$k = 0.0010$]{\includegraphics[width=.45\textwidth]{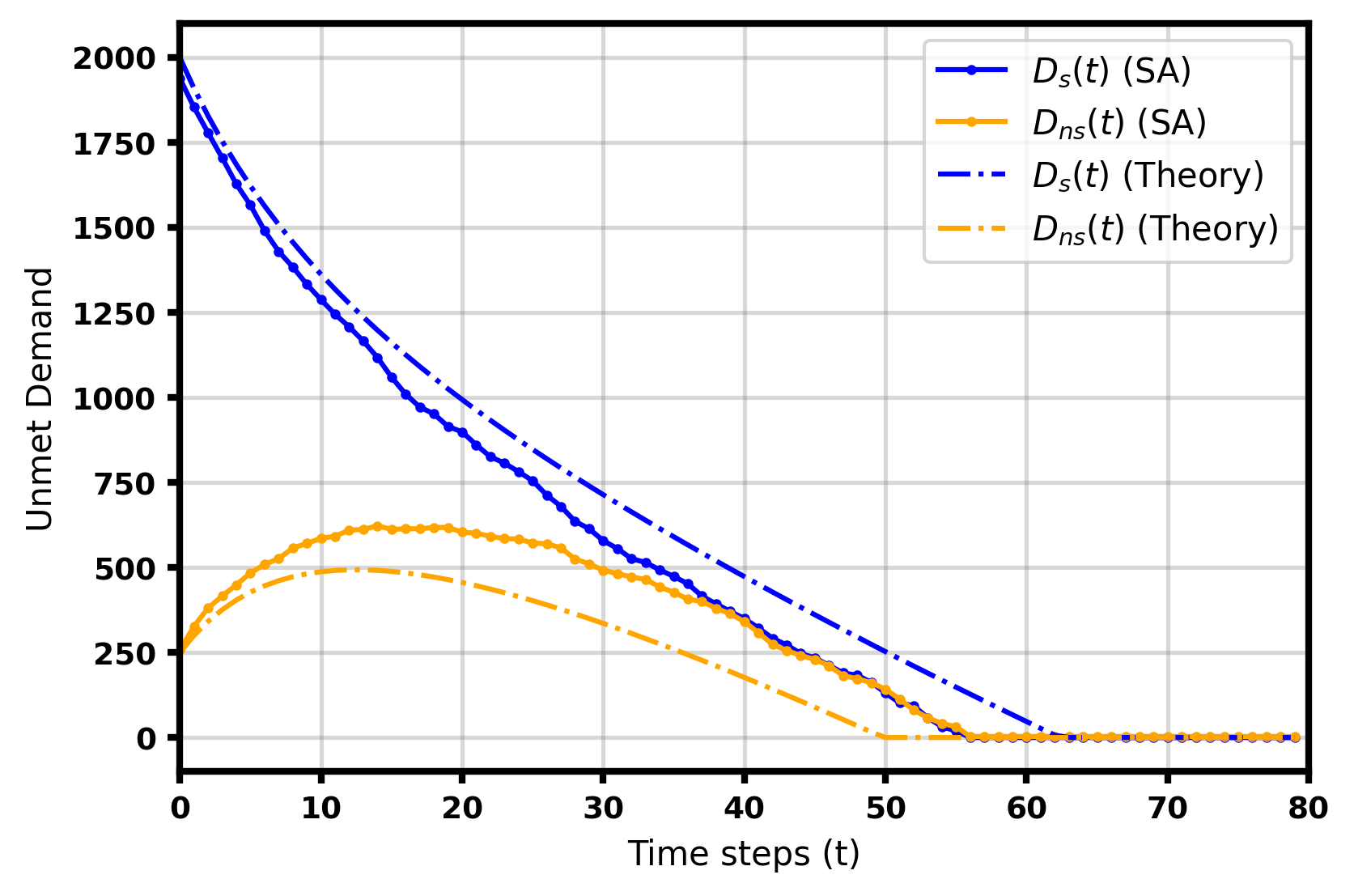}}\par
    \raisebox{35pt}{\parbox[b]{.05\textwidth}{}}%
    \subfloat[][$k = 0.0025$]{\includegraphics[width=.45\textwidth]{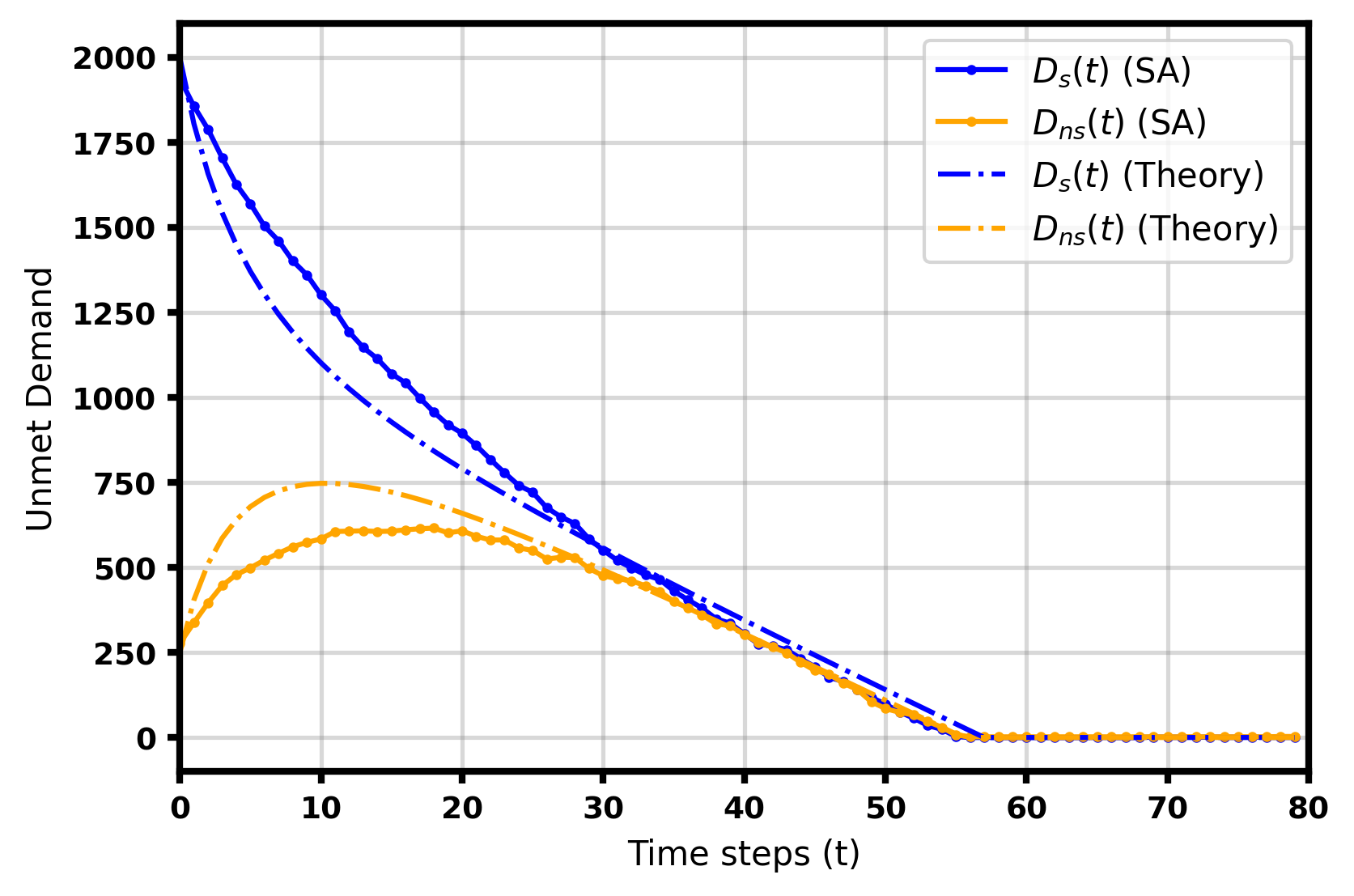}}\hfill
    \subfloat[][$k = 0.0040$]{\includegraphics[width=.45\textwidth]{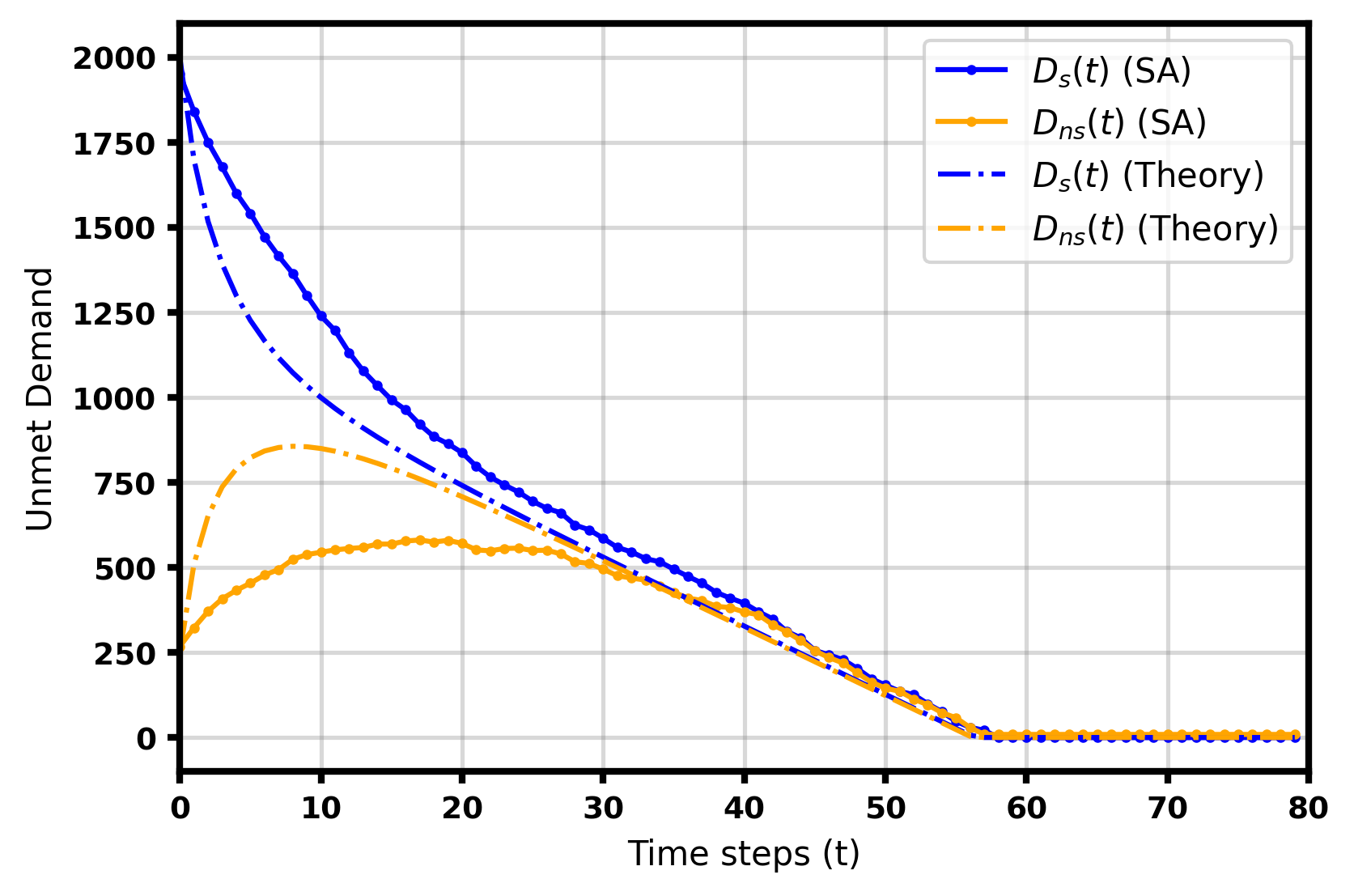}}
    \caption{Comparison of \textsf{SA} model with different instances of \textsf{Theory} model, each with a different value of $k$.}
    \label{fig:SAvsTheory_diffk}
\end{figure}
\begin{figure}
    \centering
    \includegraphics[width=0.5\linewidth]{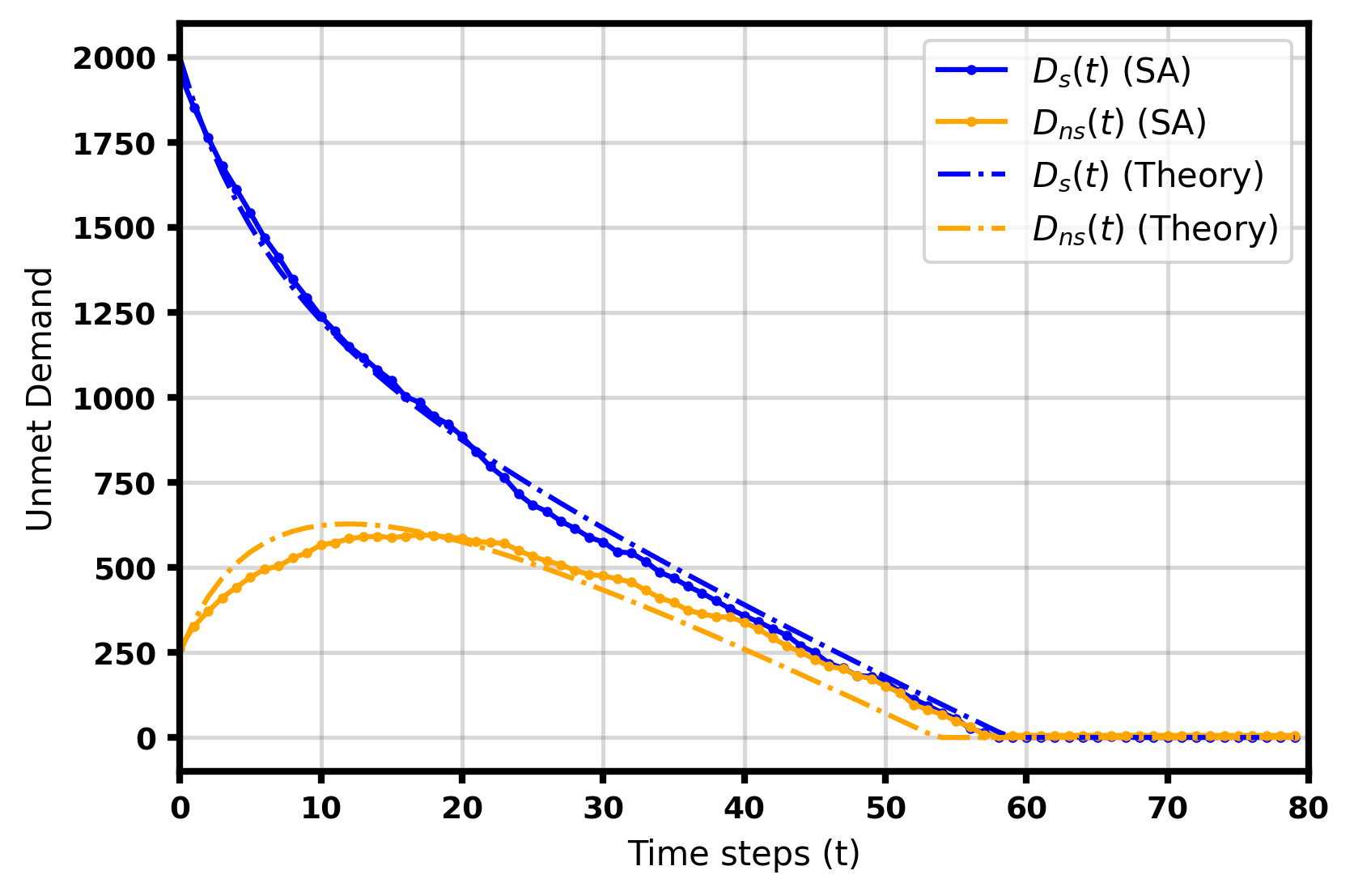}
    \caption{At $k = 0.0016$, we observe that the \textsf{Theory} model curves closely match the \textsf{SA} simulation model curves. This is for the following set of parameters: $D_0 = 2000$, $d_0 = 250$, $\lambda = 30$, $\mu = 50$, $\mathcal{D}_{mean} = 15$ and $\mathcal{D}_{std} = 8$. For any set of parameters, the optimal $k$ can be tuned.}
    \label{fig:SAvsTheory_optk}
\end{figure}

\paragraph{Comparison to Non-Strategic Benchmark Model: }We also compare the price curves and demand curves of our game-theoretic simulation model with strategic agents (\textsf{SA}) with its non-strategic counterpart (\textsf{NSB}) in Figure \ref{fig:SAvsNSB_basic}.
\begin{figure}[!ht]
    \centering
    \raisebox{35pt}{\parbox[b]{.05\textwidth}{}}%
    \subfloat[][Demand Curves]{\includegraphics[width=.48\textwidth]{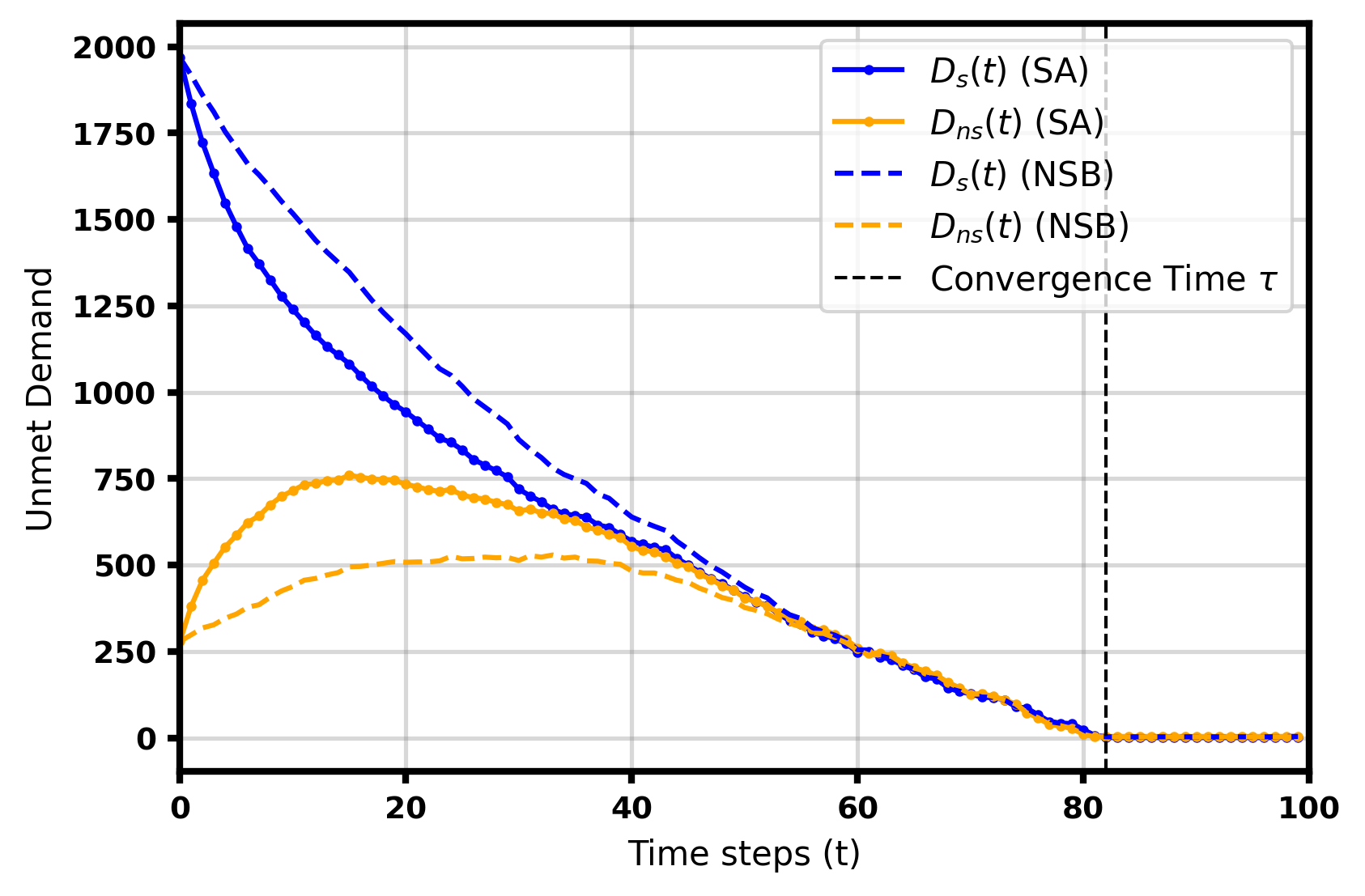}}\hfill
    \subfloat[][$\Delta P$ curves]{\includegraphics[width=.48\textwidth]{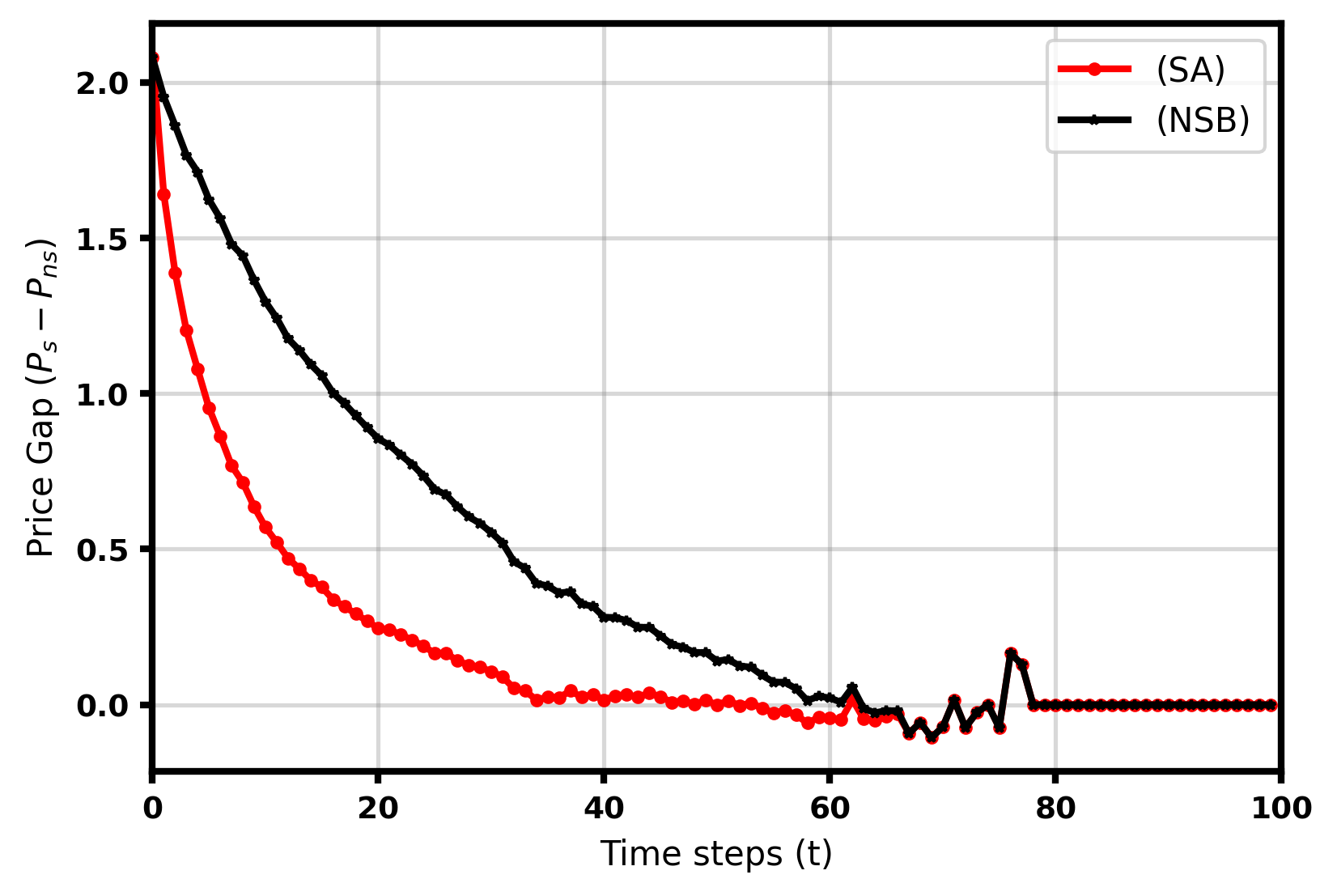}}\par
    \caption{Comparison of demand curves (a) and price curves (b) between the Strategic Agent Model (SA) and the Non-Strategic Agent Benchmark Model (NSB) with the following parameters: $D_0 = 2000$, $d_0 = 250$, $\lambda = 30$, $\mu = 45$, $\mathcal{D}_{mean} = 8$, $\mathcal{D}_{std} = 8$. Note that $\mathcal{D}$ represents the cost distribution of agents. The willingness to pay values of agents are drawn i.i.d. from a symmetric truncated normal distribution with mean $= 7$ and standard deviation $= 2$. Time to convergence $\tau = 82$ time steps.}
    \label{fig:SAvsNSB_basic}
\end{figure}
From the demand curves in (a), it is clear that the \textsf{SA} model leads to a redistribution of demand due to the effect of strategic agents moving from the surge zone to the non-surge zone. This also results in a more desirable price curve for the \textsf{SA} model with the price gap $\Delta P$ across the surge boundary diminishing and returning to normal levels faster than the \textsf{NSB} model. 

The heatmap in Figure \ref{fig:mean_deltaP} illustrates how much improvement the \textsf{SA} model achieves in terms of price gap $\Delta P$ relative to the \textsf{NSB} model for different combinations of parameters $\mathcal{D}_{mean}$ and $\mathcal{D}_{std}$. The blue regions in the heatmap represent high improvement regimes while the purple regions represent low improvement regimes. For example, there is an average 75\% improvement in \textsf{SA} price gaps compared to \textsf{NSB} for the parameter combinations used in Figure \ref{fig:SAvsNSB_basic}. Observe that the highest improvements are achieved at low values of $\mathcal{D}_{mean}$. This is because low mean values of the waiting cost distribution indicate that a large fraction of the riders have an incentive to cross the surge boundary. This leads to demand redistribution and lowering of the price gap. On the other hand, at high values of $\mathcal{D}_{mean}$ with low $\mathcal{D}_{std}$, we see the smallest levels of improvement. This is again intuitive as a high mean waiting cost coupled with a small standard deviation implies that most riders have high waiting costs and are unwilling to move. Therefore, the demand remains significantly higher in the surge zone for longer periods of time, leading to higher price gaps on average, and marginal improvements compared to the \textsf{NSB} model. Finally, at high values of $\mathcal{D}_{std}$, there is a mix of riders with low and high waiting costs due to the high variability in the distribution of costs. As a result, the price gaps become invariant of $\mathcal{D}_{mean}$ and stabilize (between $-50~\%$ and $-70~\%$) as we move towards the right of the map. 

\begin{figure}
    \centering
    \includegraphics[width=0.5\linewidth]{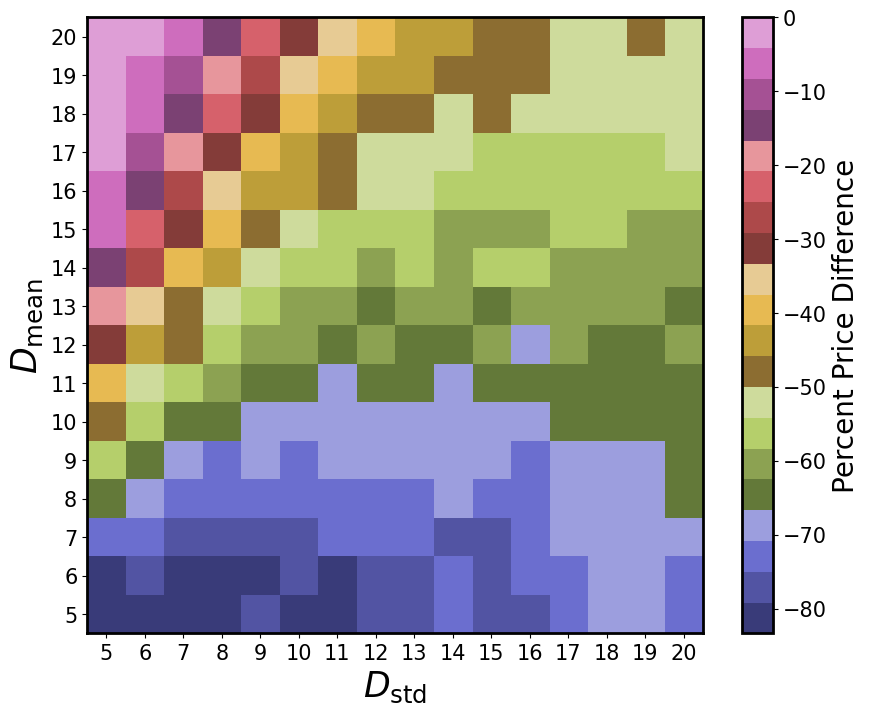}
    \caption{Relative Difference in $\Delta P$ between the SA and NSB models, given by $\frac{\Delta P_{SA} - \Delta P_{NSB}}{\Delta P_{NSB}} \times 100$, for different combinations of $\mathcal{D}_{mean}$ and $\mathcal{D}_{std}$. Reported values are time averages over $T = 500$ time steps. A more negative ratio indicates a higher degree of improvement. Demand and supply parameters are given by $D_{0} = 2000$, $d_{0}= 250$, $\mu = 45$ and $\lambda = 30$.}
    \label{fig:mean_deltaP}
\end{figure}
\paragraph{Effects of different parameters on the demand and price curves:} We now evaluate our game-theoretic simulation models for different parameter combinations where the parameters of interest are the exogenous demand and supply in each zone, the mean and scale parameters of the agents' cost distribution $\mathcal{D}$ and the initial demands $D_0$ and $d_0$. 

\begin{itemize}
\item Ratio of Supply to Exogenous Demand $\left(\frac{\mu}{\lambda}\right)$: For this set of experiments, we fix $D_0 = 2000$, $d_0 = 250$, $\mathcal{D}_{mean} = 12$, $\mathcal{D}_{std} = 8$ and $\lambda = 30$. We vary the supply rate $\mu$ to evaluate the performance of the system for different supply-to-demand ratios. 
\begin{figure}[!ht]
    \centering
    \raisebox{35pt}{\parbox[b]{.05\textwidth}{}}%
    \subfloat[][$\frac{\mu}{\lambda} = 1.25$, $\tau = 170$]{\includegraphics[width=.32\textwidth]{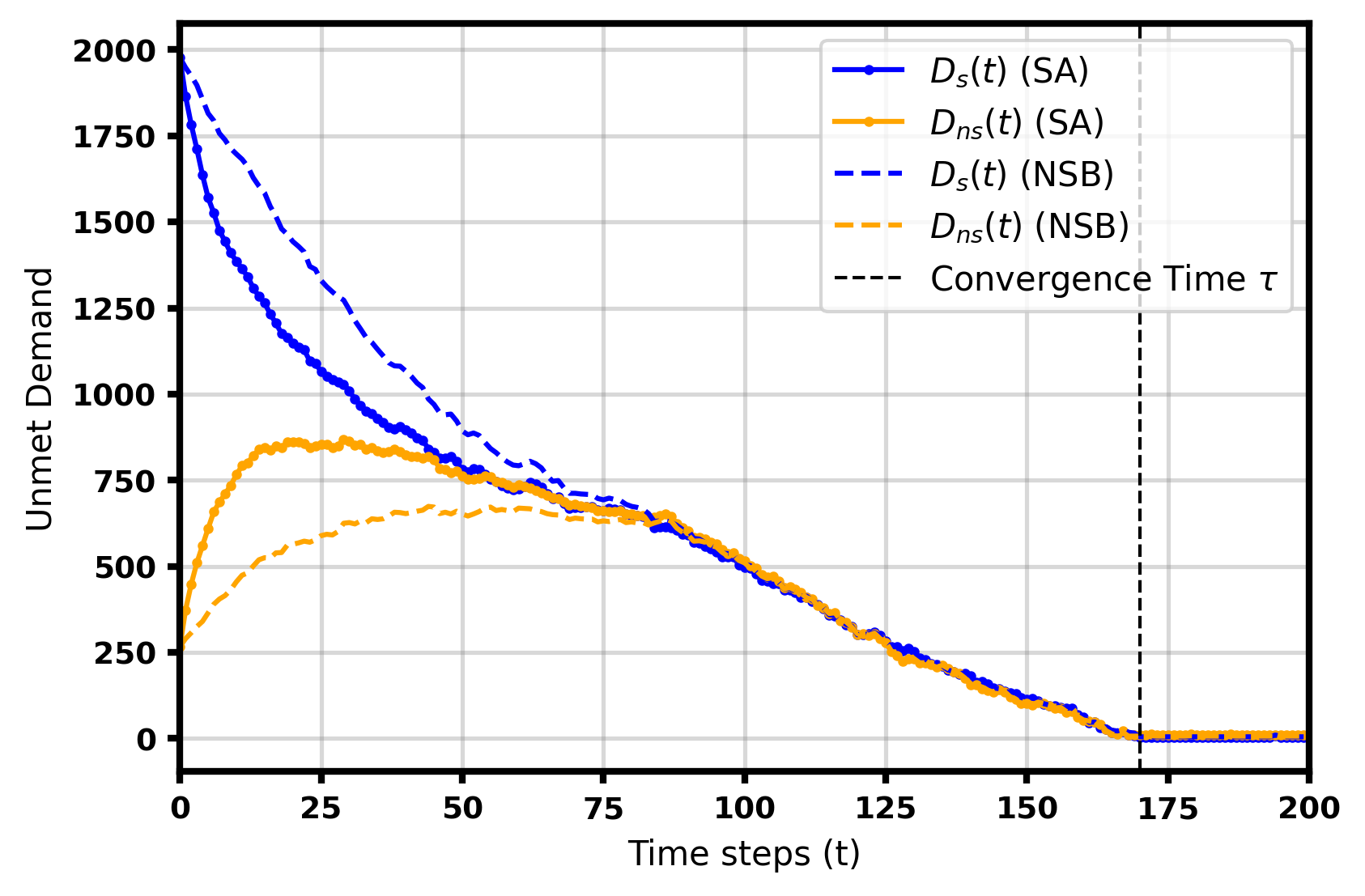}}\hfill
    \subfloat[][$\frac{\mu}{\lambda} = 1.8$, $\tau = 45$]{\includegraphics[width=.32\textwidth]{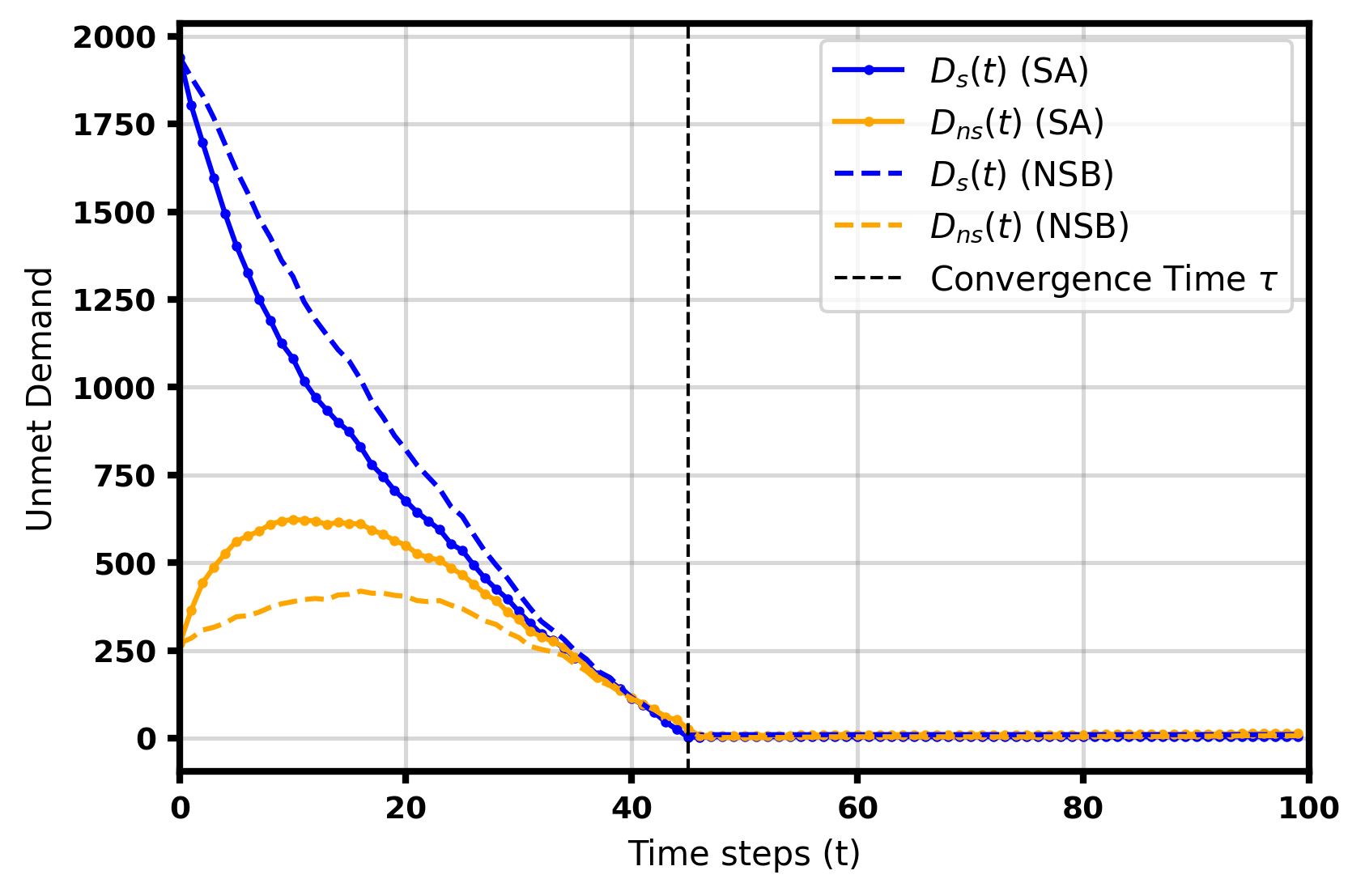}}\hfill
    \subfloat[][$\frac{\mu}{\lambda} = 2.5$, $\tau = 24$]{\includegraphics[width=.32\textwidth]{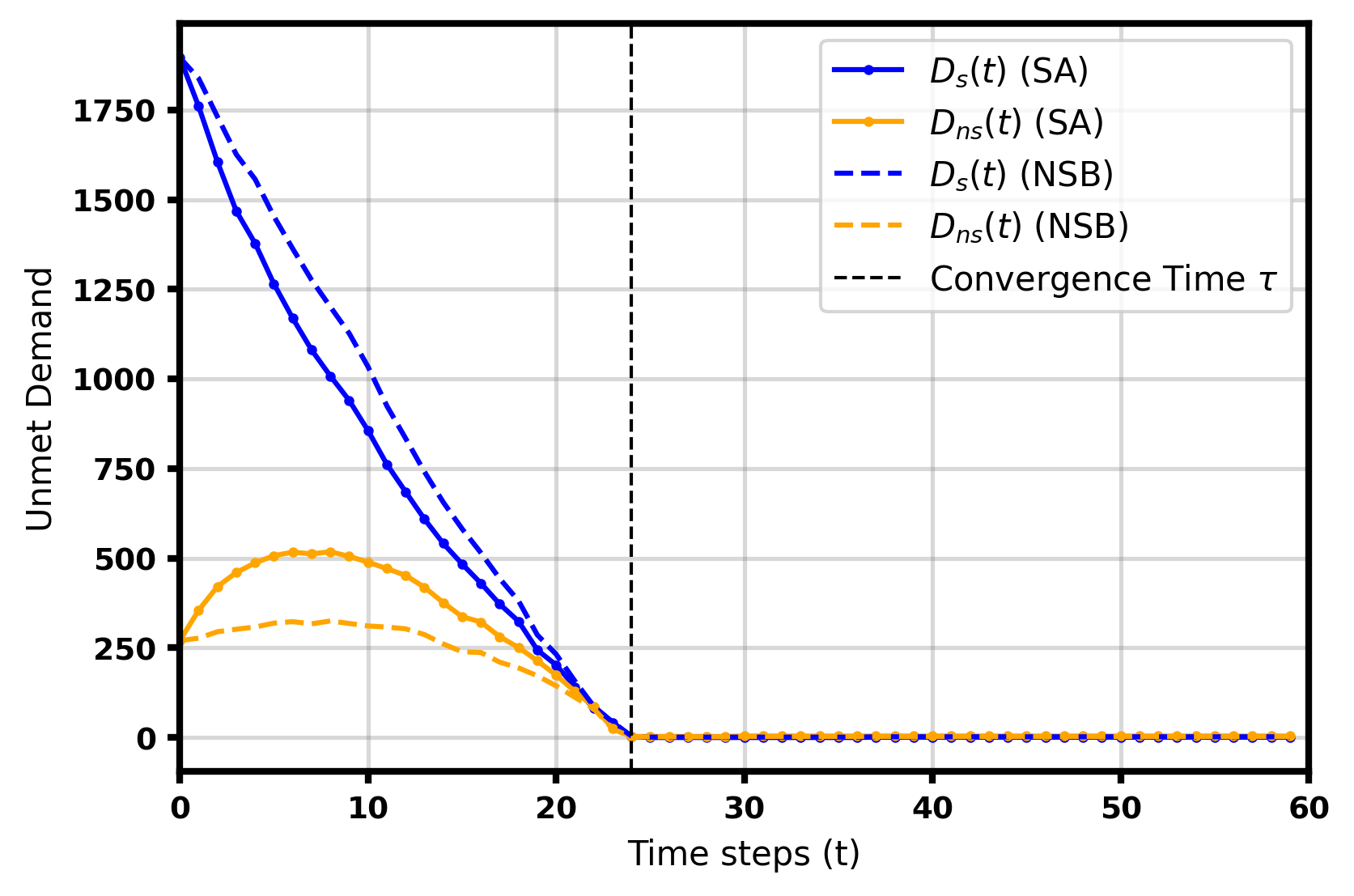}}\par
    \raisebox{35pt}{\parbox[b]{.05\textwidth}{}}%
    \subfloat[][$\frac{\mu}{\lambda} = 1.25$]{\includegraphics[width=.32\textwidth]{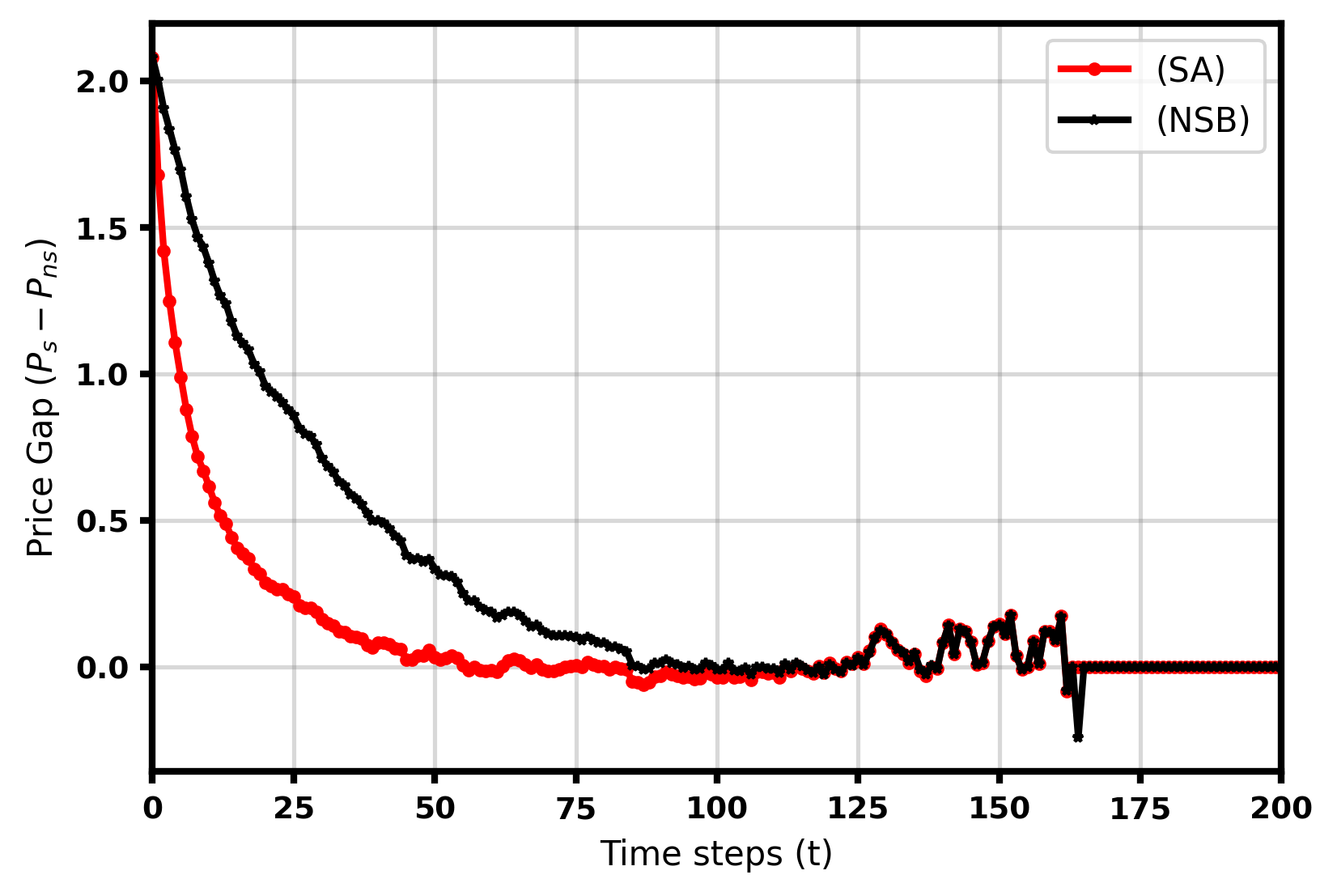}}\hfill
    \subfloat[][$\frac{\mu}{\lambda} = 1.8$]{\includegraphics[width=.32\textwidth]{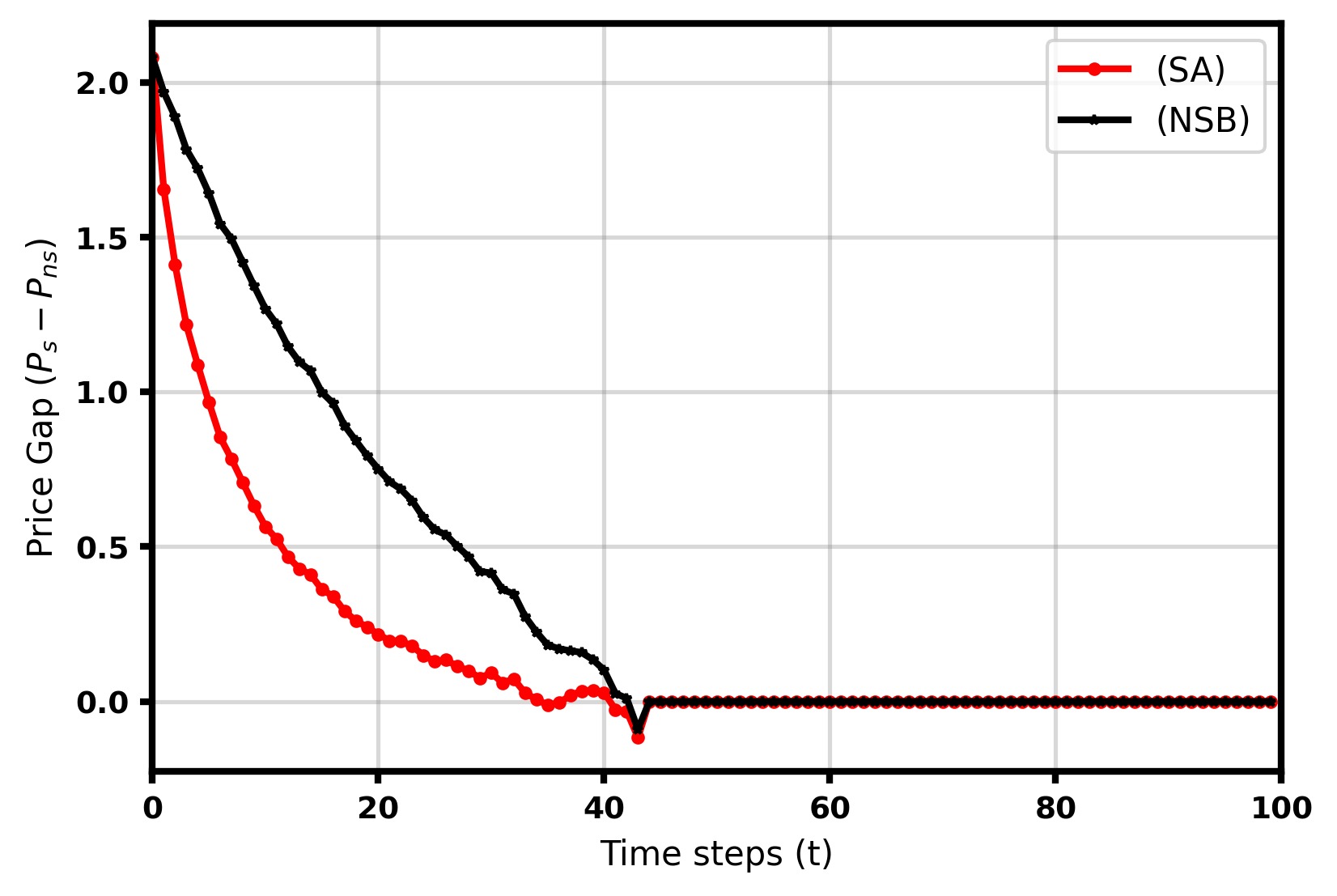}}\hfill
    \subfloat[][$\frac{\mu}{\lambda} = 2.5$]{\includegraphics[width=.32\textwidth]{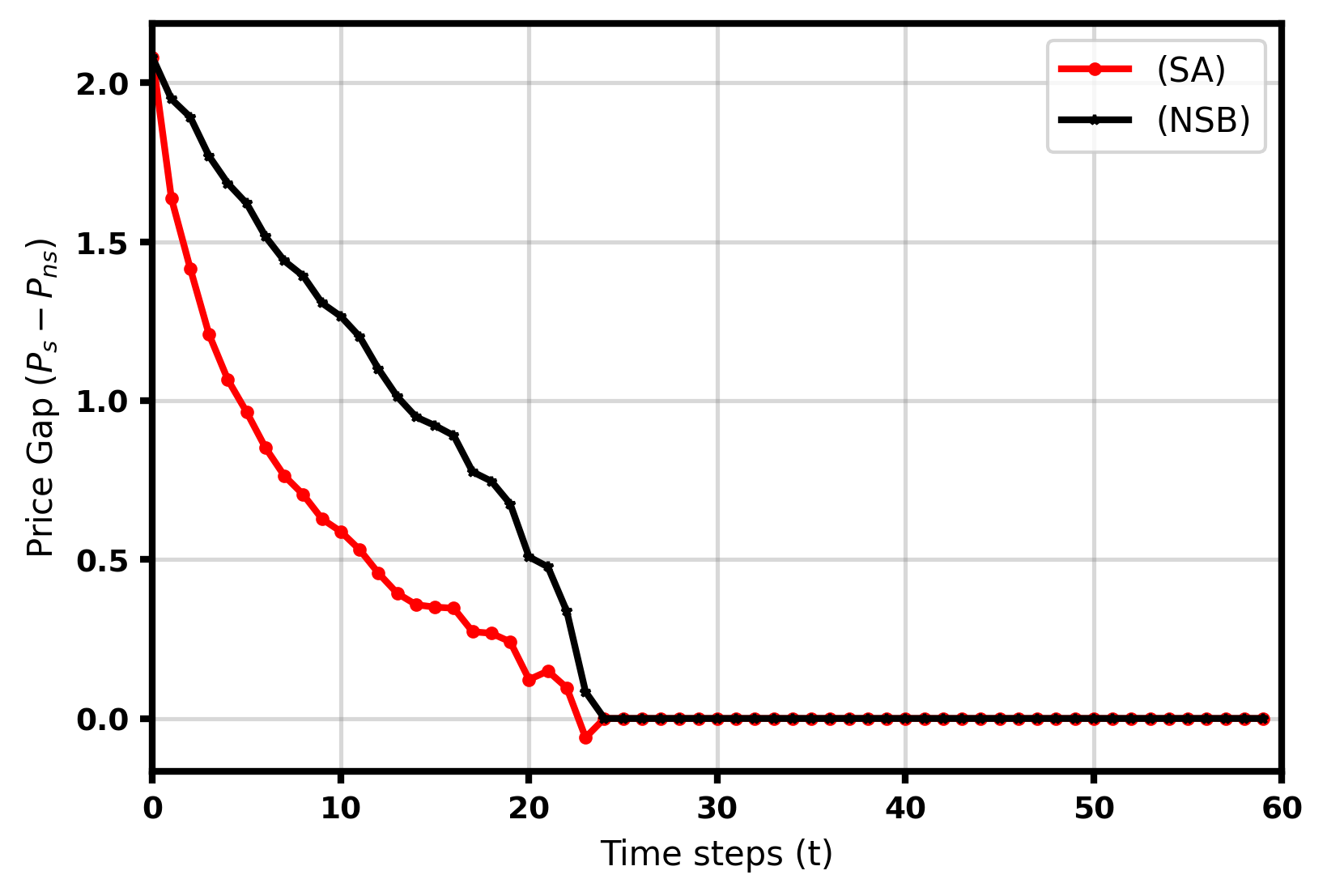}}
    \caption{Price and demand curves for \textsf{SA} and \textsf{NSB} models for different supply-demand ratios.}
    \label{fig:SAvsNSB_supply}
\end{figure}
As we increase the supply $\mu$ relative to the exogenous demand $\lambda$, the excess supply in each time period increases, leading to faster dissipation of surge and restoration of prices and unmet demands to normal levels. This can be observed by the decreasing values of $\tau$ in Figure \ref{fig:SAvsNSB_supply} as we go from left to right. 

\item \emph{Ratio of Initial Demands $D_0$ and $d_0$:} For this set of experiments, we fix $d_0 = 250$, $\mathcal{D}_{mean} = 12$, $\mathcal{D}_{std} = 8$ and $\lambda = 30$ and $\mu = 60$. $D_0$ is varied to capture different relative sizes of uncleared demands in the two zones. The ratio of $D_0$ to $d_0$ represents the magnitude of the demand disparity caused by the surge. As $D_0$ increases, the maximum $\Delta P$ is found to increase. This is because the large demand disparity forces the platform to set much higher prices in the surge zone to capture a larger proportion of the total supply and still equalize waiting time across the two zones. Additionally, since $\mu$ and $\lambda$ remain unchanged, the total demand clearing rate in the system is still the same but there is larger demand to be cleared. So, the system takes much longer to dissipate the surge, i.e., $\tau$ increases from left to right. Refer to Figure \ref{fig:SAvsNSB_demand}.
\begin{figure}[!ht]
    \centering
    \raisebox{35pt}{\parbox[b]{.05\textwidth}{}}%
    \subfloat[][$\frac{D_0}{d_0} = 5$, $\tau = 24$]{\includegraphics[width=.32\textwidth]{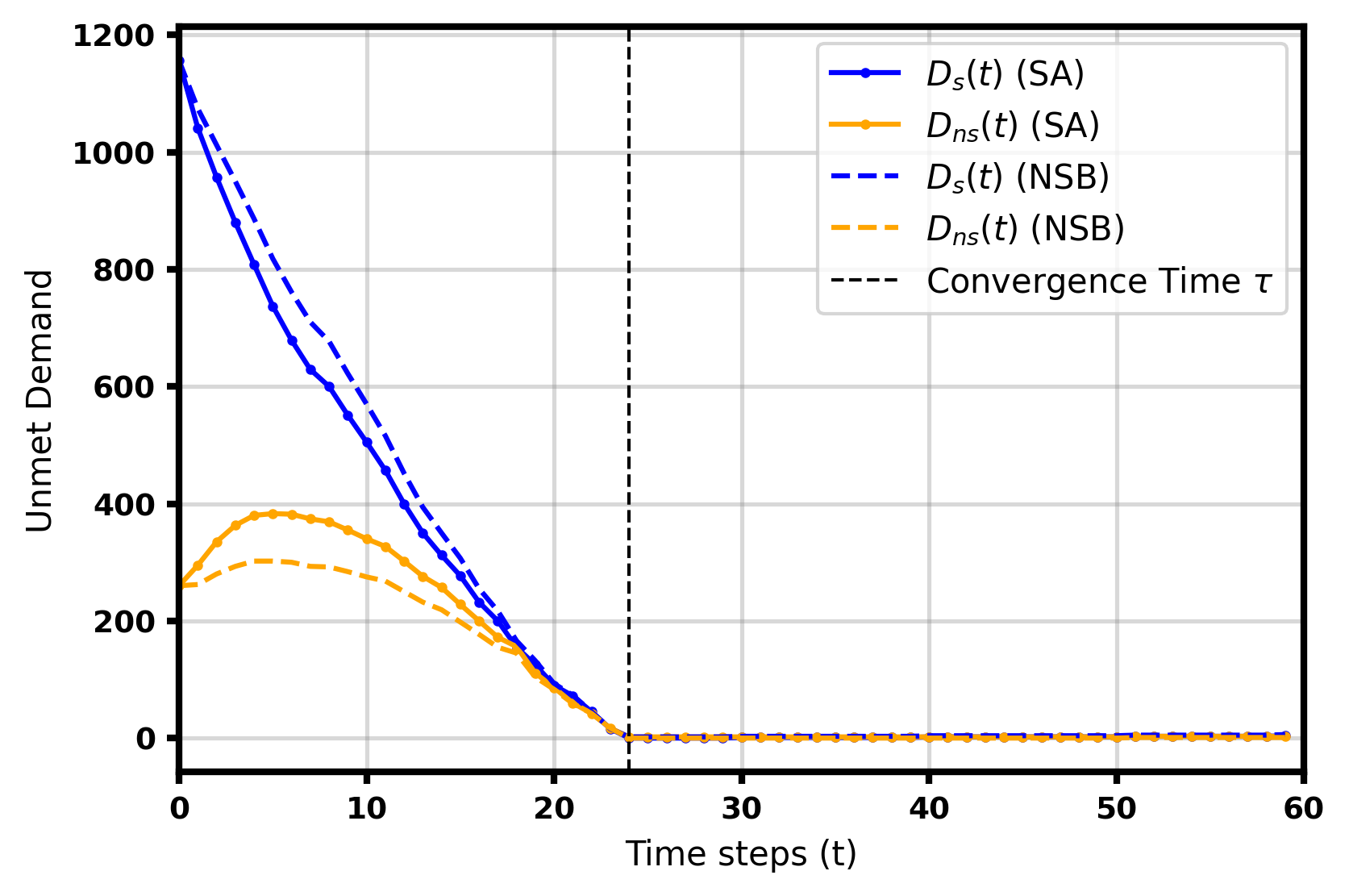}}\hfill
    \subfloat[][$\frac{D_0}{d_0} = 10$, $\tau = 47$]{\includegraphics[width=.32\textwidth]{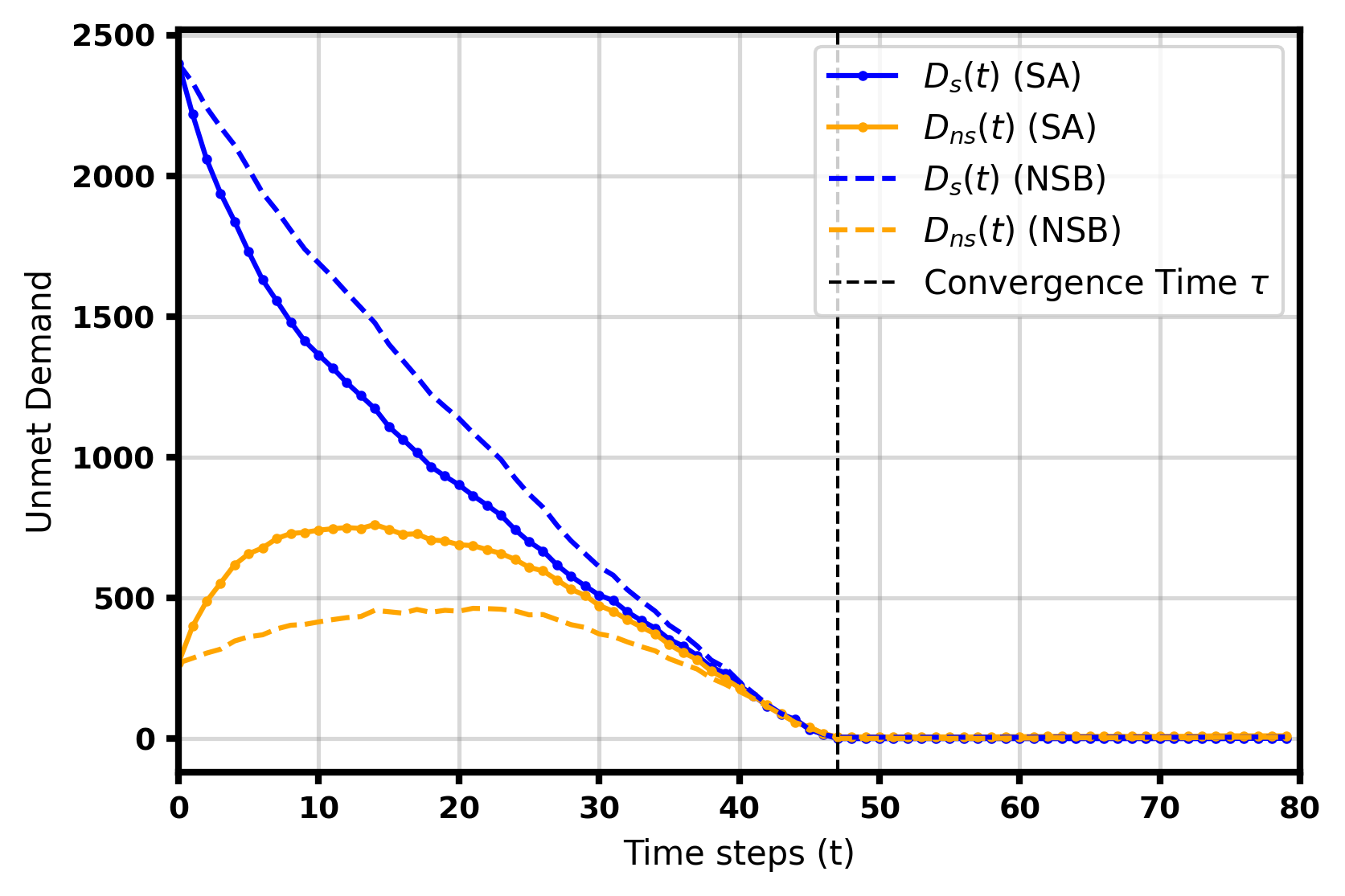}}\hfill
    \subfloat[][$\frac{D_0}{d_0} = 20$, $\tau = 91$]{\includegraphics[width=.32\textwidth]{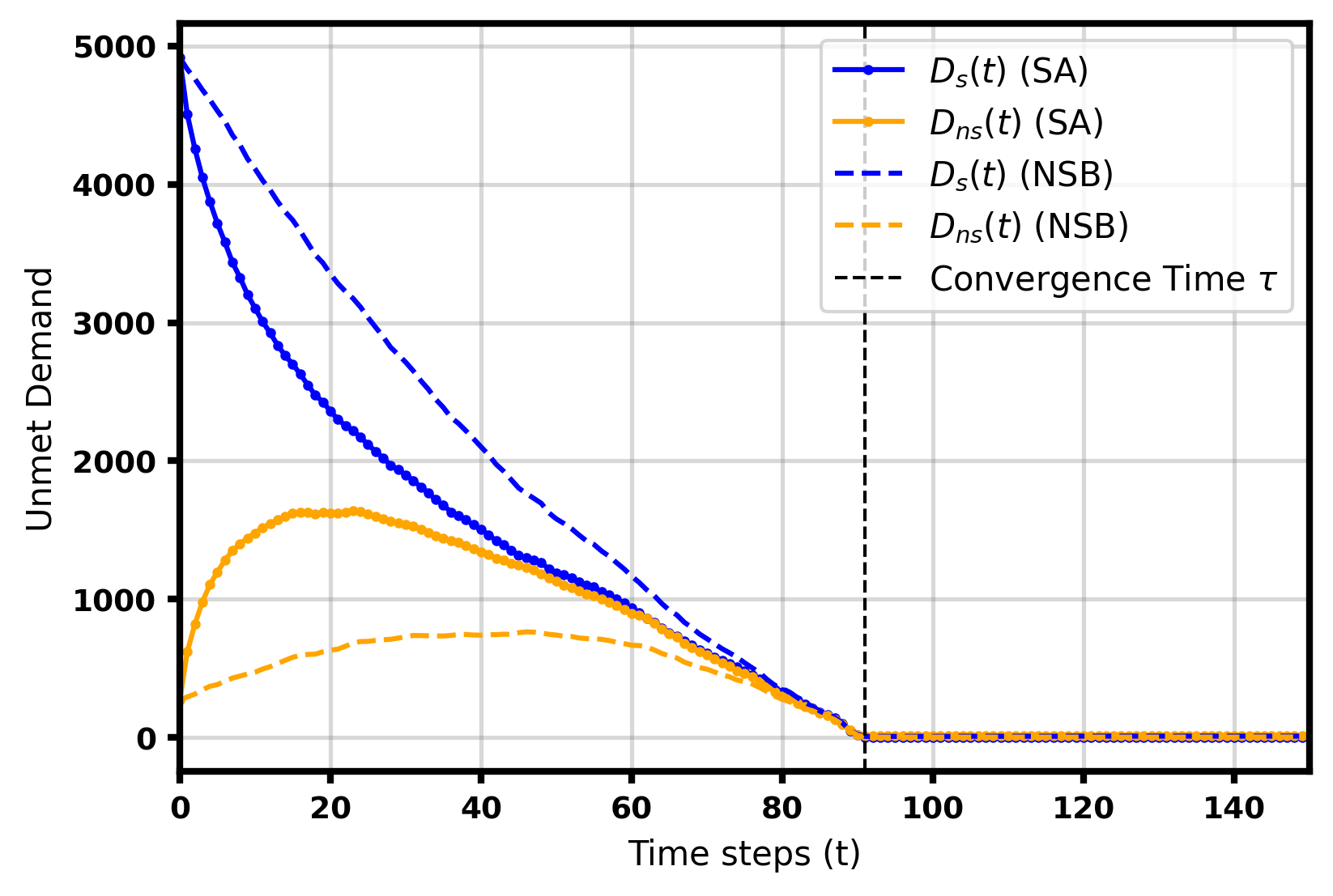}}\par
    \raisebox{35pt}{\parbox[b]{.05\textwidth}{}}%
    \subfloat[][$\frac{D_0}{d_0} = 5$]{\includegraphics[width=.32\textwidth]{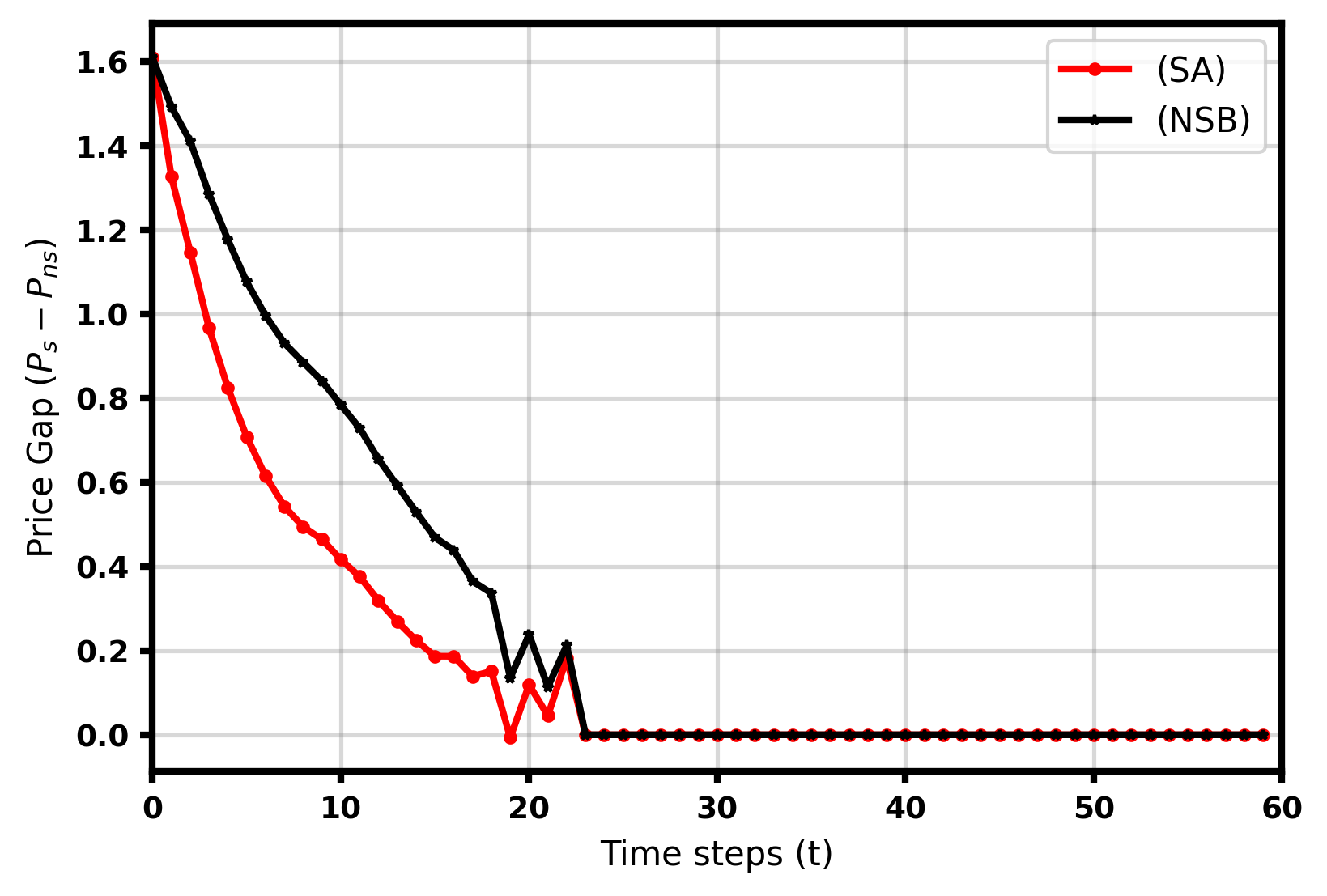}}\hfill
    \subfloat[][$\frac{D_0}{d_0} = 10$]{\includegraphics[width=.32\textwidth]{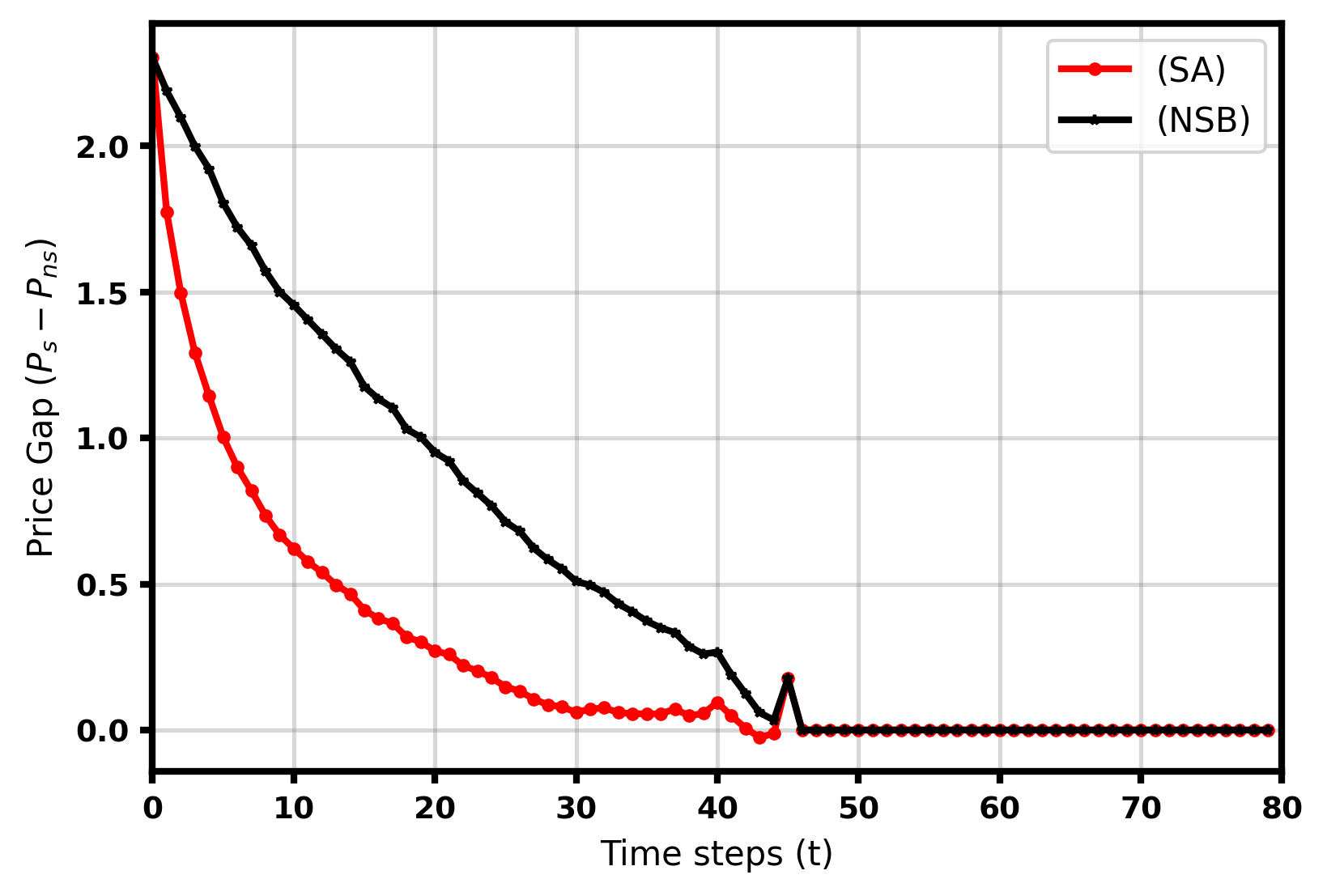}}\hfill
    \subfloat[][$\frac{D_0}{d_0} = 20$]{\includegraphics[width=.32\textwidth]{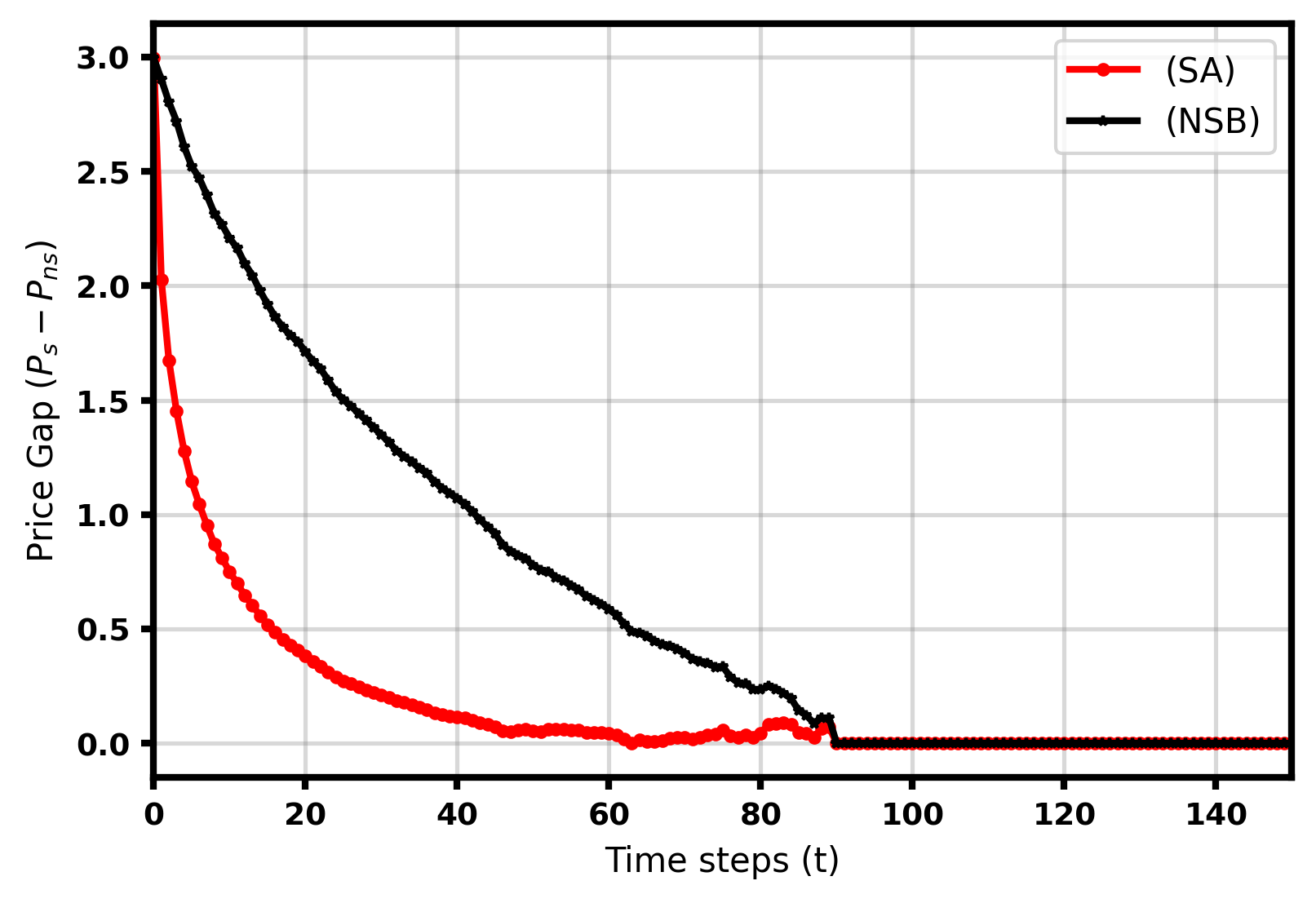}}
    \caption{Price and demand curves for \textsf{SA} and \textsf{NSB} models for different ratios of initial demands.}
    \label{fig:SAvsNSB_demand}
\end{figure}

\item \emph{Effect of the mean parameter of distribution $\mathcal{D}$:}
We can vary parameters of agents' cost distribution $\mathcal{D}$ keeping all other system parameters fixed with values: $D_0 = 2000$, $d_0 = 250$, $\mu = 60$, $\lambda = 30$. Recall that the parameters of $\mathcal{D}$ regulate what fraction of agents would choose to relocate in each time period as a function of the price gap $\Delta P$. For the results in Figure \ref{fig:SAvsNSB_d}, we fix $\mathcal{D}_{std}$ = 5 and vary the $\mathcal{D}_{mean}$. As $\mathcal{D}_{mean}$ increases, we see a lesser number of strategic agents leaving the surge zone, causing the \textsf{SA} and \textsf{NSB} curves to approach each other as we move further to the right. Note that changing the parameters of $\mathcal{D}$ has no impact on the convergence times $\tau$, as expected.  
\begin{figure}[!ht]
    \centering
    \raisebox{35pt}{\parbox[b]{.05\textwidth}{}}%
    \subfloat[][$\mathcal{D}_{mean} = 5$, $\tau = 36$]{\includegraphics[width=.32\textwidth]{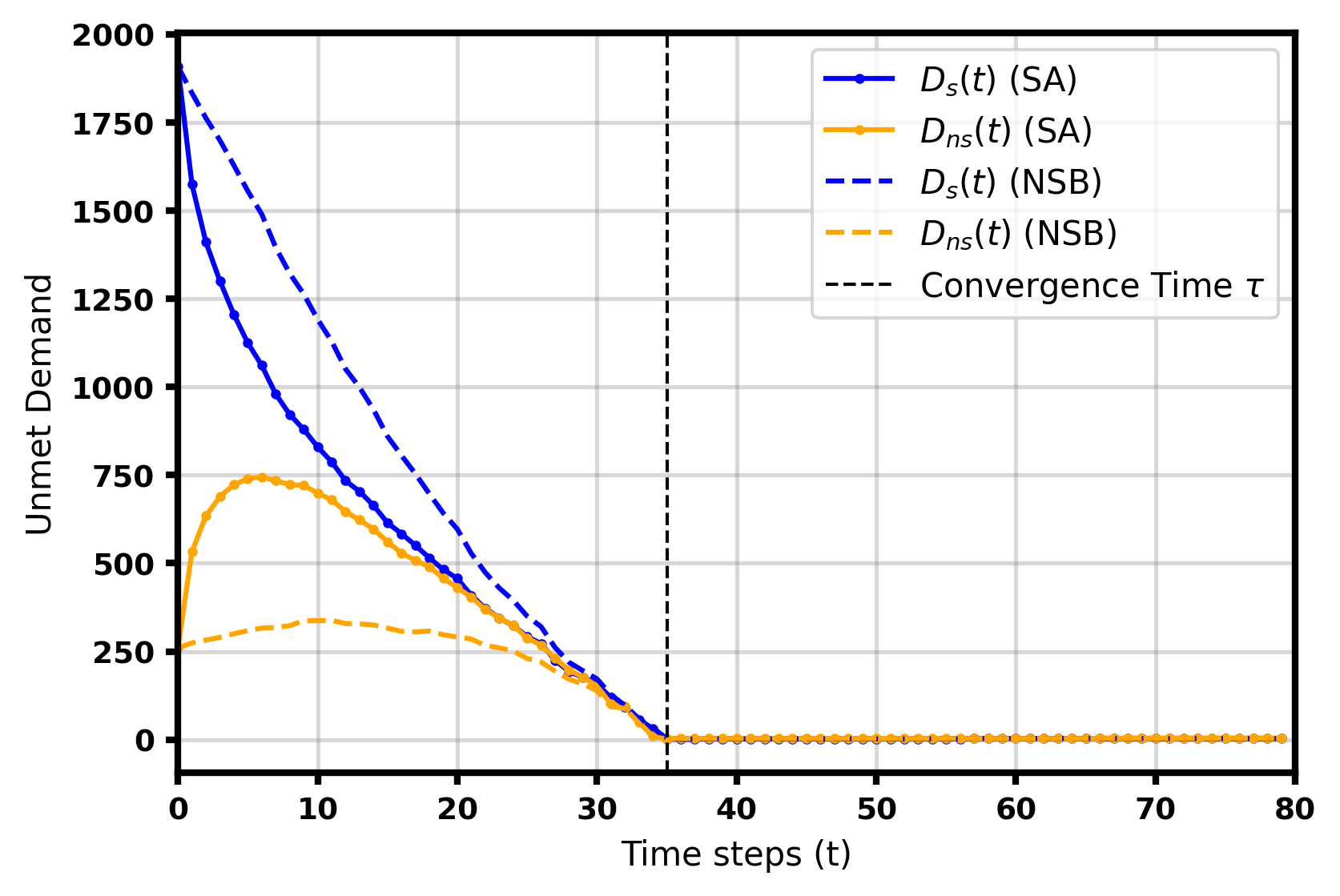}}\hfill
    \subfloat[][$\mathcal{D}_{mean} = 10$, $\tau = 37$]{\includegraphics[width=.32\textwidth]{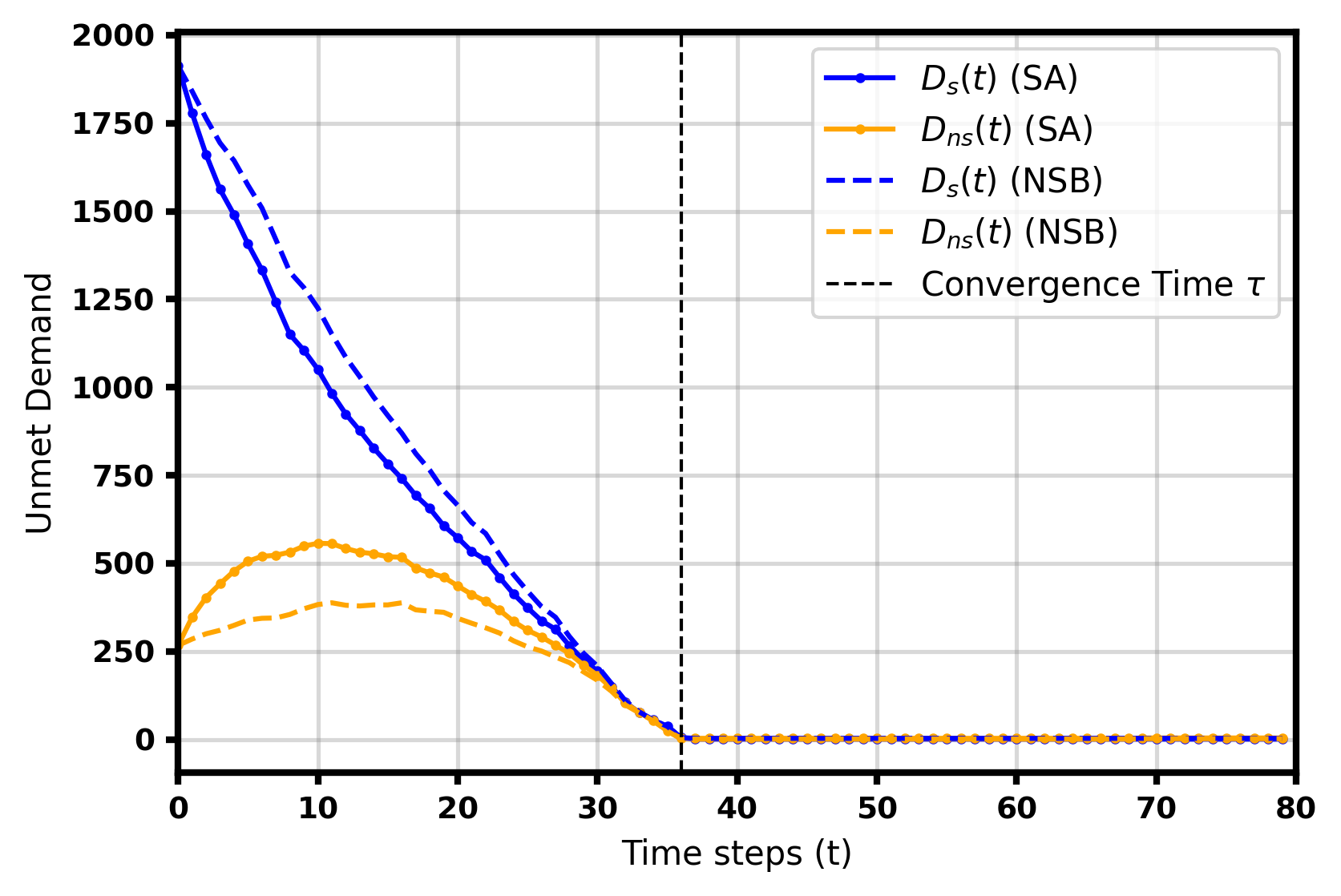}}\hfill
    \subfloat[][$\mathcal{D}_{mean} = 15$, $\tau = 36$]{\includegraphics[width=.32\textwidth]{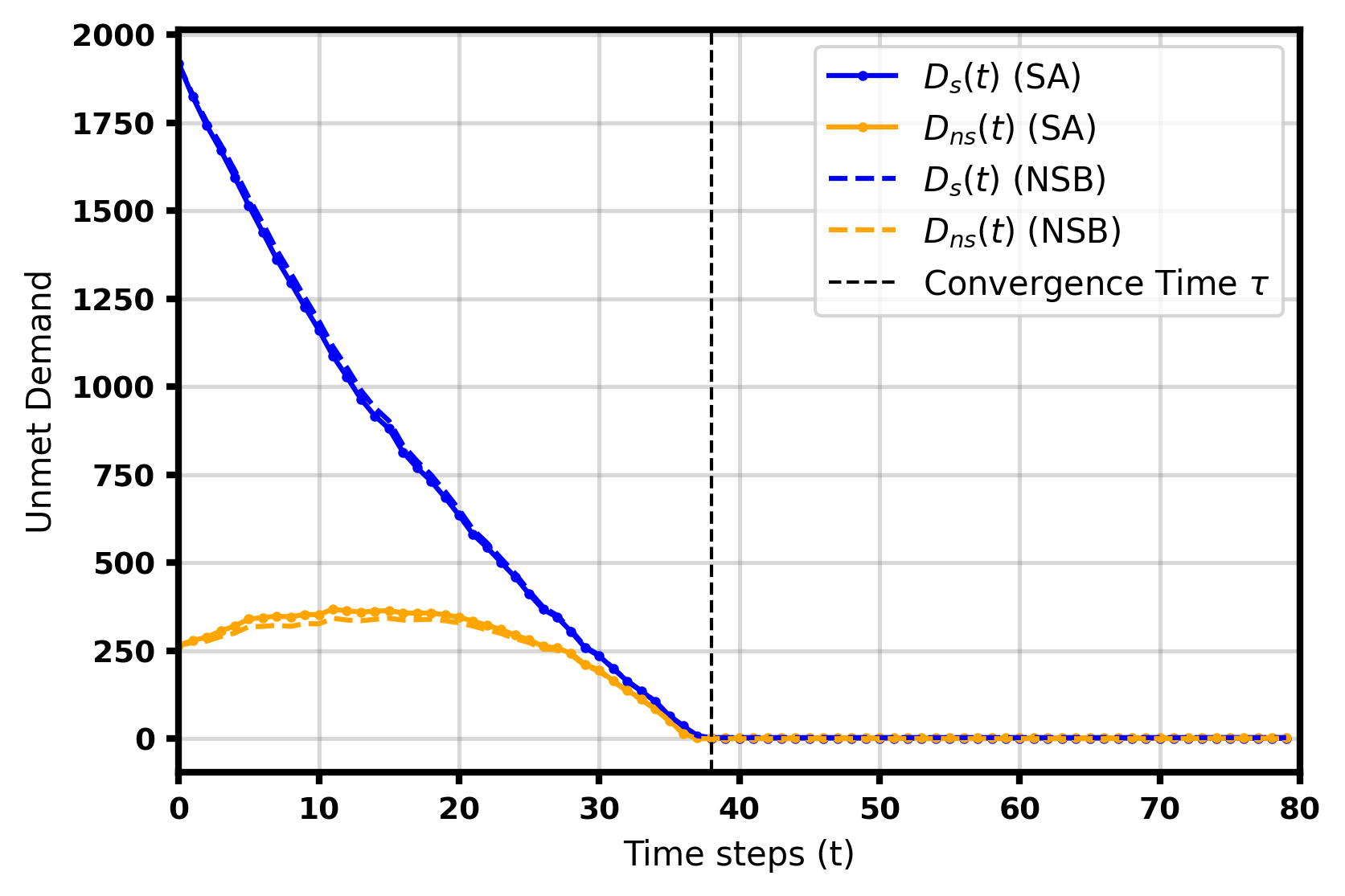}}\par
    \raisebox{35pt}{\parbox[b]{.05\textwidth}{}}%
    \subfloat[][$\mathcal{D}_{mean} = 5$]{\includegraphics[width=.32\textwidth]{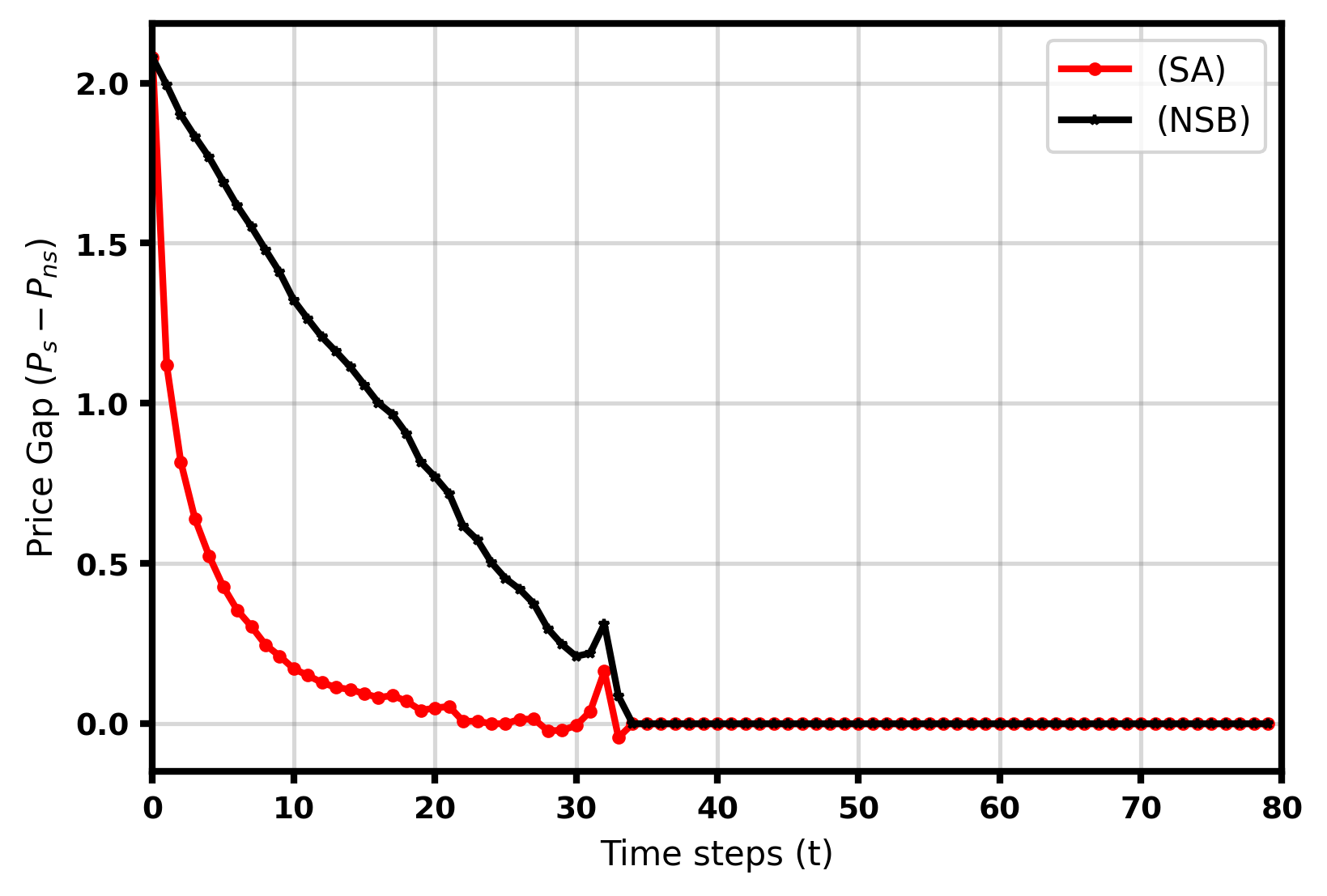}}\hfill
    \subfloat[][$\mathcal{D}_{mean} = 10$]{\includegraphics[width=.32\textwidth]{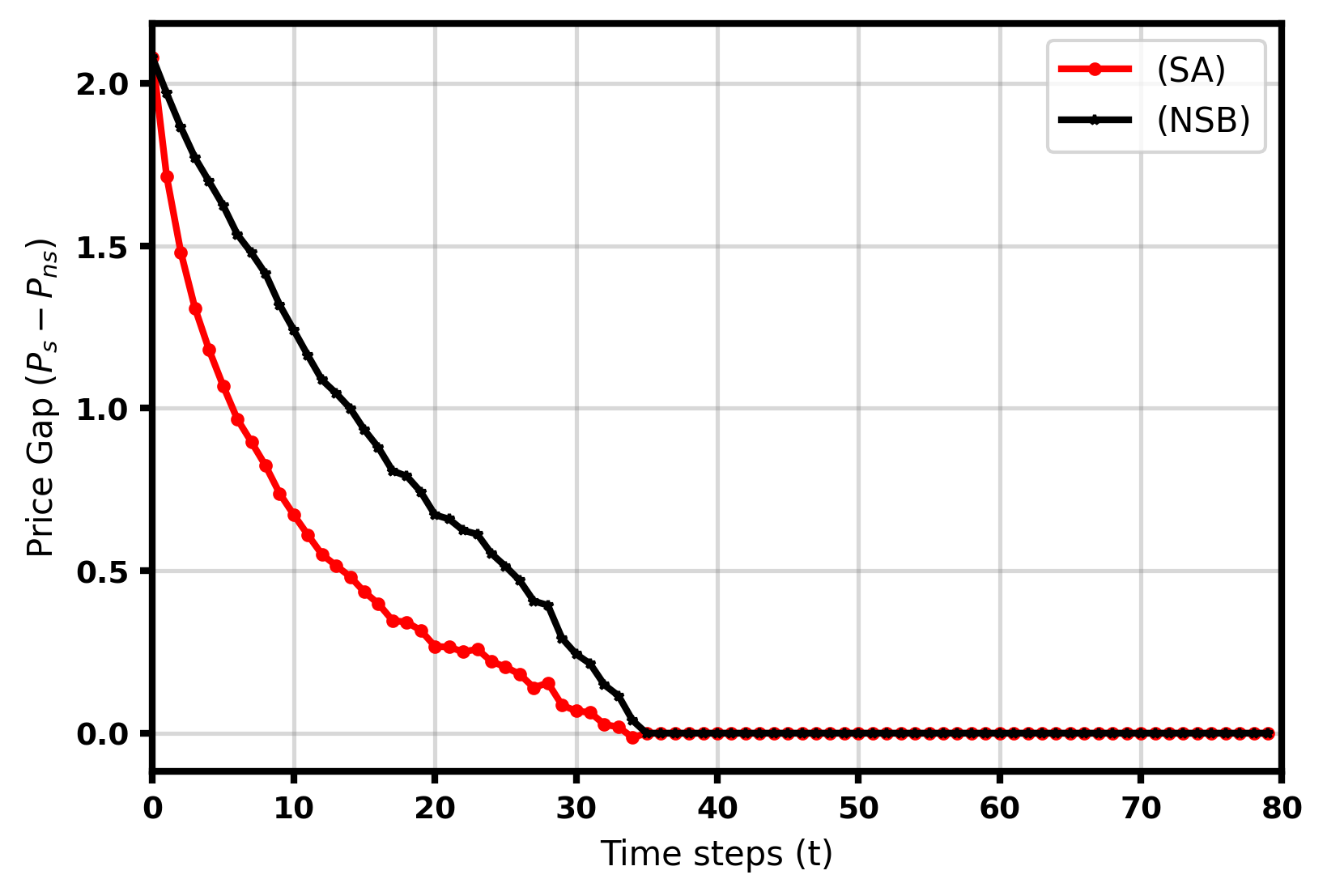}}\hfill
    \subfloat[][$\mathcal{D}_{mean} = 15$]{\includegraphics[width=.32\textwidth]{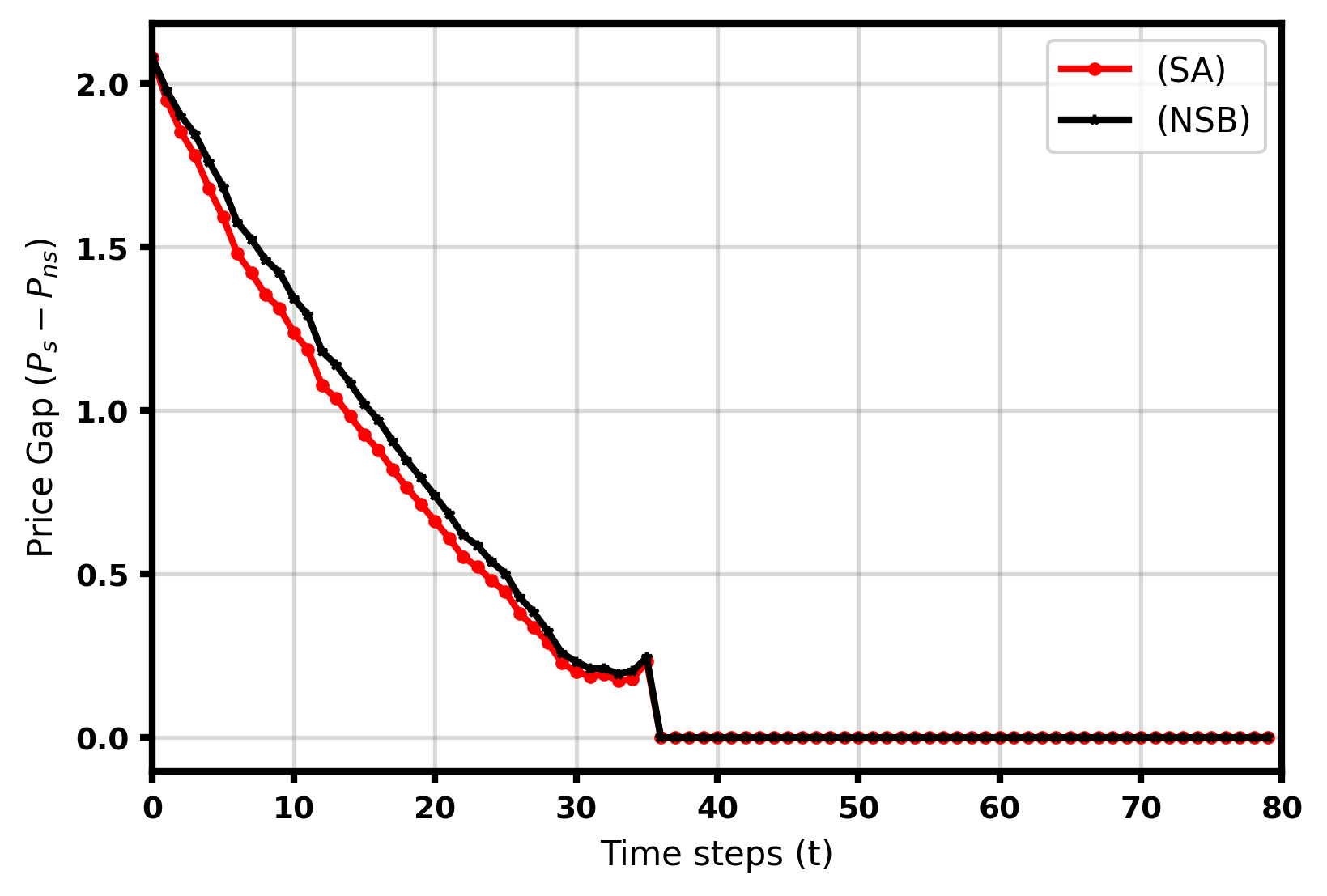}}
    \caption{Price and demand curves for \textsf{SA} and \textsf{NSB} models for varying $\mathcal{D}_{mean}$ }
    \label{fig:SAvsNSB_d}
\end{figure}
\end{itemize}

\section{Discussion \& Future Work}\label{sec:disc}
In this work, we have explored an important aspect of strategic behavior of riders during price surges on ride-sharing platforms---riders monitor prices in real-time and may choose to walk away from the surge zone in search of cheaper rides. We have shown analytically how such strategic behavior has the potential to redistribute demand spatially and through extensive numerical experiments, demonstrated how it can equalize prices across the surge boundary faster. However, there are many interesting avenues of future work. 

In recent years, we have seen an uptick in ride-sharing platforms related to riders deciding to walk when faced with high prices and long waiting times. Platforms have been trying to directly and explicitly incentivize riders to walk to and from pick-up and drop-off points by providing them with cheaper rides if they do so. Uber's ExpressPool~\citep{uber_expresspool} and Lyft's Shared Saver~\citep{lyft_walksave} are key examples of such a ``walk and save'' feature. In doing so, the platform can reduce driver downtime while improving operational efficiency. Over the years, these features have been modified and re-introduced multiple times which demonstrates an ongoing rider-side demand for such an option. This raises an important mechanism design question---how should platforms design rider incentives in an optimal way so that they can be encouraged to avail them while still protecting platform profits? 

In addition to leaving the surge zone, another important aspect of strategic behavior on the rider side is waiting for prices to drop. Earlier works have modeled this effect in isolation; in the real world, however, strategic rider behavior is often a combination of both aspects: riders often initially wait at the surge zone for prices to drop, and then consider walking to other locations after too long of a wait. One immediate extension of our work is to model this more complex rider-side behavior with both waiting and walking effects simultaneously.

\bibliographystyle{plainnat}
\bibliography{mybib.bib}

\newpage

\appendix
\section{Auxillary Results \& Missing Proofs}\label{app:proofs}

\subsection{Convergence}
\begin{theorem}\label{auxthm:total_conv}
The total unmet demand in the system, given by $D_s(t)+D_{ns}(t)$, must converge to $0$ in finite time. 
\end{theorem}
\begin{proof}
$D_s(t) \geq 0$ and $D_{ns}(t) \geq 0$ for all $t$. Therefore, $D_s(t) + D_{ns}(t) > 0$ if and only if at least one of $D_s(t)$ and $D_{ns}(t)$ is strictly greater than $0$. Therefore, in each time period $t$, at least a $\mu - \lambda > 0$ amount (this amount is in fact $2(\mu-\lambda) > \mu-\lambda > 0$ if both demands are positive) of 
unmet demand can be met every time step. This implies that $D_s(t) + D_{ns}(t)$ must reach $0$ in at most $\left \lceil \frac{D_0+d_0}{\mu-\lambda} \right \rceil$ time steps. 
\end{proof}

\begin{corr}\label{auxcorr:D_conv}
Each of $D_s(t)$ and $D_{ns}(t)$ must converge to $0$ in finite time. 
\end{corr}

\begin{proof}
This is a direct consequence of the previous theorem, since both demands are non-negative. 
\end{proof}

\subsection{Shape Properties}
\begin{theorem}\label{auxthm:Ds_conv}
The unmet demand in the surge zone, given by $D_s(t)$, is non-increasing in $t$ for all $t \in \{0, 1, 2,....\}$ and must reach $0$ in at most $\left \lceil \frac{D_0}{\mu - \lambda} \right \rceil$ time steps.
\end{theorem}
\begin{proof}
Recall that we have the following update rule for $D_s(t)$:
\[
   D_s(t+1) = \max \left[0, D_s(t) - (\mu - \lambda) - f\left( D_s(t)-D_{ns}(t)\right)\cdot D_s(t) \right].
\]
Since $\mu > \lambda$ (by assumption), $D_s(t) = 0$ implies: 
\[
      D_s(t+1) = \max\left[0, -(\mu-\lambda) \right] = 0.
\]
By induction, we can conclude that $D_s(t+2) = D_s(t+3) = .... = 0$. Thus, once $D_s(t)$ reaches $0$, it remains at $0$ and hence is non-increasing for all future time points. 

Now, suppose $D_s(t+1) > 0$. Then, by Equation ~\ref{eq:evolve_s_rev}, we must have:
\begin{align*}
    D_s(t+1) = D_s(t) - (\mu - \lambda) - f \left( D_s(t)-D_{ns}(t) \right)\cdot D_s(t) \leq D_s(t) - (\mu - \lambda).
\end{align*}
The inequality follows because $f(\cdot) \geq 0$ (by definition). 
Thus we have established that $D_s(t)$ is monotonically decreasing until it reaches $0$, where each step (except for the last one) registers a decrease by at least $(\mu-\lambda)$.    
Since $D_s(0) = D_0$, this implies $D_s(t)$ must reach $0$ in at most $\left \lceil \frac{D_0}{\mu-\lambda} \right \rceil$ steps. This concludes the proof. 
\end{proof}

\begin{theorem}\label{auxthm:Dns}
Let $D_{ns}(t)$ be the unmet demand in the non-surge zone at any time step $t \in \{0, 1, 2,...\}$. There always exists some time step $\tau \geq 0$ such that $D_{ns}(t)$ is non-increasing for all $t \geq \tau$. $\tau$ is defined as follows:
\[
   \tau = \min \left\{t \in \{0, 1, 2,....\}: f\left(D_s(t)-D_{ns}(t) \right)\cdot D_s(t) \leq \mu - \lambda \right\}.
\]
Further, $D_{ns}(t)$ is non-decreasing for all $t < \tau$. 
\end{theorem}
\begin{proof}
We will complete the proof in two steps. First, we will use the earlier result that $D_{ns}(t)$ must converge to $0$ in finite time to argue about the existence of $\tau$. Finally, we will show that $D_{ns}(t)$ is non-increasing beyond $\tau$. 
\end{proof}

\begin{claim}
There exists a finite $\tau \in \{0, 1, 2...\}$ such that: 
\[
    \tau = \min \left\{t \in \{0, 1, 2,....\}: f\left(D_s(t)-D_{ns}(t) \right)\cdot D_s(t) \leq \mu - \lambda \right\}.
\]
Further, $D_{ns}(t)$ is increasing for all $t < \tau$.
\end{claim}
\begin{proof}
We will prove by contradiction. Suppose, such a $\tau$ does not exist. This implies, for all $t \in \{0, 1, 2,....\}$, it must be that: 
\[
    f\left(D_s(t)-D_{ns}(t) \right)\cdot D_s(t) > \mu - \lambda,
\]
which implies that:
\[
     (\lambda-\mu) + f\left(D_s(t)-D_{ns}(t) \right)\cdot D_s(t) > 0 \quad \forall~ t.
\]
Therefore, using Equation ~\ref{eq:evolve_ns_rev}, we conclude that if $D_{ns}(t) > 0$, then:
\[
    D_{ns}(t+1) = \max \left[0, D_{ns}(t) +  (\lambda-\mu) + f\left(D_s(t)-D_{ns}(t) \right)\cdot D_s(t) \right] >  D_{ns}(t).
\]
Since $D_{ns}(0) = d_0 > 0$, $\left \{D_{ns}(t) \right \}_{t}$ must constitute a positive, strictly increasing sequence. But this contradicts our earlier result that $D_{ns}(t)$ must converge to $0$; hence $\tau$ must exist. 

Finally since $\tau$ exists, from the definition of $\tau$, it must be that for all $t < \tau$: 
\[
    f\left(D_s(t)-D_{ns}(t) \right)\cdot D_s(t) > \mu - \lambda
\]
Following a similar argument as earlier, we conclude that $D_{ns}(t)$ must be increasing for all $t < \tau$. This concludes the proof.  
\end{proof}

\begin{claim}
$D_{ns}(t)$ is non-increasing for all $t \geq \tau$. 
\end{claim}
\begin{proof}
In order to prove that $D_{ns}(t)$ is always non-increasing beyond $\tau$, note that it is sufficient to show that:
\begin{align}\label{temp}
    f \left( D_s(t)-D_{ns}(t) \right)\cdot D_s(t) \leq (\mu - \lambda) \quad \forall ~t \geq \tau.
\end{align}
We will do a proof by induction: \\
\emph{Base Case: }From the definition of $\tau$ in the previous claim, we already know that \eqref{temp} holds at $t = \tau$. This implies $D_{ns}(\tau+1) \leq D_{ns}(\tau)$. \\
\emph{Induction Hypothesis:} Suppose, \eqref{temp} holds for some $t' > \tau$.\\
\emph{General Case: } We now need to show that \eqref{temp} also holds for $t'+1$. There are two cases to be considered:
\begin{itemize}
    \item Suppose, $D_s(t') = 0$. In that case, by Theorem \ref{auxthm:Ds_conv}, $D_s(t'+1) = 0$. This implies:
    \[
          f \left( D_s(t'+1)-D_{ns}(t'+1) \right)\cdot D_s(t'+1) = 0 \leq (\mu - \lambda).
    \]

    \item Now, suppose, $D_s(t') > 0$. If $D_s(t'+1) = 0$, we are done again (trivially). If not, i.e., $D_s(t'+1) > 0$, we have:
    \begin{align*}
        D_s(t'+1) - D_{ns}(t'+1)
        &= D_s(t') + (\lambda - \mu) - f\left( D_s(t')-D_{ns}(t') \right) \cdot D_s(t')\\
        & \quad - \max \left[0, D_{ns}(t') + (\lambda-\mu) + f\left( D_s(t')-D_{ns}(t')\right)\cdot D_s(t') \right] \\
        &\leq  D_s(t') + (\lambda - \mu) - f\left( D_s(t')-D_{ns}(t') \right) \cdot D_s(t') \\
        & \quad - \left( D_{ns}(t') + (\lambda-\mu) + f\left( D_s(t')-D_{ns}(t')\right)\cdot D_s(t') \right)\\
        &=  D_s(t') - D_{ns}(t') - f\left( D_s(t')-D_{ns}(t')\right)\cdot D_s(t') \\
        &\leq D_s(t') - D_{ns}(t').
    \end{align*}
\end{itemize}
The first step follows because $D_s(t'+1) > 0$ by assumption. The inequality in the second step results from the observation that removing the $\max(\cdot)$ can only increase the difference, with exact equality holding if $D_{ns}(t'+1) > 0$. The last step follows since $f(\cdot) \geq 0$ and $D_s(t') \geq 0$. Therefore, by the non-decreasing property of $f(\cdot)$, we must have:
\[
   f\left( D_s(t'+1) - D_{ns}(t'+1) \right) \leq f \left( D_s(t') - D_{ns}(t') \right).
\]
Finally, since $D_s(t'+1) \leq D_s(t')$ by Theorem \ref{auxthm:Ds_conv}, we establish the desired inequality and conclude the proof.
\[
   f\left( D_s(t'+1) - D_{ns}(t'+1) \right) \cdot D_s(t'+1) \leq f \left( D_s(t') - D_{ns}(t') \right)\cdot D_s(t') \leq \mu - \lambda. 
\]
\end{proof}

\subsection{Bounds on Convergence Times}
\begin{theorem}\label{auxthm:Ds_bounds}
The time to convergence for $D_s(t)$, given by $\tau_s$, must satisfy:
\[
        \left \lfloor \frac{D_0}{(\mu-\lambda) + D_0 \cdot f(D_0)} \right \rfloor \leq \tau_s \leq \left \lceil \frac{D_0}{\mu - \lambda} \right \rceil.
\]
\end{theorem}
\begin{proof}
We have already established in the proof of Theorem B.3 that $\tau_s \leq \left \lceil \frac{D_0}{(\mu - \lambda)} \right \rceil$. In the rest of the proof, we only need to establish the lower bound. The update rule for $D_s(t)$ is given by: 
\begin{align*}
    D_s(t+1) &= D_s(t) - (\mu - \lambda) - f\left( D_s(t)-D_{ns}(t) \right)\cdot D_s(t) \\
    &\geq D_s(t) - (\mu - \lambda) - f\left( D_s(t) \right)\cdot D_s(t)
\end{align*}
The inequality follows because $D_s(t) - D_{ns}(t) \leq D_s(t)$ and $f(\cdot)$ is non-decreasing in its argument by definition. 
Now, using the structure of the update rule, we can adopt a telescopic sum approach: 
\begin{align*}
    D_s(\tau_s) - D_s(\tau_s-1) &\geq - (\mu - \lambda) - f\left(D_s(\tau_s-1) \right) \cdot D_s(\tau_s-1) \\
    &......\\
    D_s(1) - D_s(0) &\geq - (\mu - \lambda) - f\left(D_s(0) \right) \cdot D_s(0).
\end{align*}
Adding up all equations and substituting $D_s(0) = D_0$, we obtain: 
\[
    D_s(\tau_s) - D_0 \geq -\tau_s(\mu - \lambda) - \sum_{t=0}^{\tau_s-1} f\left(D_s(t) \right) \cdot D_s(t)
\]
Note that we also have $D_s(\tau_s) = 0$ (by definition of $\tau_s$. Finally, flipping the inequality (due to the negative sign), we have:
\[
     D_0 \leq \tau_s(\mu - \lambda) + \sum_{t=0}^{\tau_s-1}f\left( D_s(t) \right)\cdot D_s(t) \leq \tau_s(\mu - \lambda) + \tau_s f\left( D_0\right)\cdot D_0
\]
This results in the desired lower bound on $\tau_s$ and concludes the proof:
\[
    \tau_s \geq \frac{D_0}{(\mu - \lambda) + D_0 \cdot f(D_0)} \geq \left \lfloor \frac{D_0}{(\mu - \lambda) + D_0 \cdot f(D_0)} \right\rfloor.
\]
\end{proof}

\begin{theorem}\label{auxthm:Dns_bounds}
The time to convergence for $D_{ns}(t)$, given by $\tau_{ns}$, must satisfy:
\[
      \left \lfloor \frac{d_0}{\mu-\lambda} \right \rfloor \leq \tau_{ns} \leq \left \lceil \frac{D_0+d_0}{\mu - \lambda} \right \rceil.
\]
\end{theorem}
\begin{proof}
The proof is straightforward. The upper bound follows from Theorem \ref{auxthm:total_conv} which shows that the total unmet demand $D_s(t)+D_{ns}(t)$ must converge in at most $\left \lceil \frac{D_0+d_0}{\mu-\lambda}\right \rceil$ time steps. Clearly, $D_{ns}(t)$ cannot take longer than those many time steps to converge. The lower bound is obtained by noting that in the absence of strategic riders, $D_{ns}(t)$ would converge in $\left \lceil \frac{d_0}{\mu-\lambda}\right \rceil$ steps starting at $d_0$. So, with additional demand from relocating riders from the surge zone, $D_{ns}(t)$ cannot converge faster.    
\end{proof}

\end{document}